\documentclass[11pt]{article}
\usepackage{float}
\usepackage{authblk}
\usepackage{parskip}
\usepackage[normalem]{ulem} 
\usepackage{fullpage}
\usepackage{url}
\usepackage{verbatim}
\usepackage{amstext}
\usepackage{amsmath,amssymb}
\usepackage{amsfonts}
\usepackage{mathtools}
\usepackage{float}
\usepackage{tikz}
\usetikzlibrary {arrows.meta} 
\usepackage{subcaption} 
\usepackage{color}
\usepackage{amsthm}
\usepackage{dsfont}
\usepackage[continuous]{causets}
\usepackage{amssymb}
\usepackage{amsmath}
\usepackage{causets}
\usetikzlibrary {graphs,quotes}
\usetikzlibrary{shapes.geometric}
\newtheorem{theorem}{Theorem}[section]
\newtheorem*{definition}{Definition}
\newtheorem{lemma}[theorem]{Lemma} 
\newtheorem{corollary}[theorem]{Corollary}

\newcommand{\bC}{\mathbb C}
\newcommand{\cH}{\mathcal H}

\newcommand{\hA}{\hat { A}}
\newcommand{\ket}[1]{| #1 \rangle}
\newcommand{\bra}[1]{\langle #1|}
\newcommand{\braket}[2]{\langle #1| #2\rangle}
\newcommand{\cyl}[1]{{\mathrm{cyl}({#1})}}
\newcommand{\tc}{{\tilde c}}
\newcommand{\tcn}{{\tc}_n}
\newcommand{\cP}{\mathcal P} 
\newcommand{\tcP}{\tilde \cP}
\newcommand{\mA}{\mathfrak A}
\newcommand{\mS}{\mathfrak S} 
\newcommand{\cZ}{\mathcal Z}
 
\newcommand{\cT}{\mathcal T}

\newcommand{\hG}{\hat G}
 
\newcommand{\hQ}{\hat Q}

\newcommand{\cR}{\mathcal R}
\newcommand{\hB}{\hat B}
\newcommand{\hS}{\hat S}

\newcommand{\rt}{\rightarrow}

\renewcommand{\a}{\alpha}

\renewcommand{\a }{\alpha}

\def\bit{\begin{itemize}}
\def\eit{\end{itemize}}
\usepackage{enumerate}

\usetikzlibrary{fit,shapes.geometric}
  \colorlet{future colour}{green!50!black}

  \tikzset{prob arrow/.style={line width=0.5pt}}
  \def\defCevents#1#2#3{\xdef\Ea{#1}\xdef\Eb{#2}\xdef\Ec{#3}}
  \def\defDevents#1#2#3#4{\expandafter\defCevents#1#2#3\xdef\Ed{#4}}
  
  \def\defEvents#1#2#3#4#5{\expandafter\defDevents#1#2#3#4\xdef\Ee{#5}}
  \newcommand*{\semiopaque}[1]{
 \begin{scope}[transparency group, opacity=0.5]
  #1
  \end{scope}
  }
\newcommand*{\drawprobarrow}[4][]{
  \draw[prob arrow] (#2)-- node[sloped,midway, below, #1] {$#3$} (#4);
}
\newcommand*{\drawprobleft}[4][]{
  \draw[prob arrow] (#2)-- node[above left, below, midway, #1] {$#3$} (#4);
}
\newcommand*{\drawprobright}[4][]{
  \draw[prob arrow] (#2)-- node[above right,below, midway, #1] {$#3$} (#4);
}
\newcommand*{\drawprobrighttop}[4][]{
  \draw[prob arrow] (#2)-- node[right=0.2,midway, #1] {$#3$} (#4);
}
\newcommand*{\drawproblefttop}[4][]{
  \draw[prob arrow] (#2)-- node[left=0.2,midway, #1] {$#3$} (#4);
  }
  \newcommand{\drawlegendsymbol}{
  \draw[prob arrow] (0,0)-- +(10pt, 6pt);
  }
  \newcommand{\legendsymbol}[1]{
  \begin{tikzpicture}
  \ifnum#1=0
  \drawlegendsymbol
  \else
  \semiopaque{\drawlegendsymbol}
  \fi
  \end{tikzpicture}
 }
 
\newcommand{\fA}{\mathfrak A}  
\newcommand{\fS}{\mathfrak S}
\newcommand{\cA}{\mathcal A}  
\newcommand{\CSG}{\bC S G}
\newcommand{\bN}{\mathbb N}
\newcommand{\chain}{\mathbf c}
\newcommand{\achain}{a}
\newcommand{\bchain}{b}

\newcommand{\cI}{\mathcal I}
\newcommand{\cJ}{\mathcal J}
\newcommand{\cK}{\mathcal K} 
\newcommand{\Past}{\mathrm{Past}}
\newcommand{\Fut}{\mathrm{Fut}}
\newcommand{\IPast}{\mathrm{IPast}}
\newcommand{\IFut}{\mathrm{IFut}} 
\newcommand{\gmu}{\mu_{\mathbf{g}}}
\newcommand{\pre}{\mathrm{Prec}}
\newcommand{\spec}{\mathrm{Sp}}
\newcommand{\cspec}{\mathrm{CSpec}}
\newcommand{\setzeta}{\Gamma}
\newcommand{\re}{\mathbb R}
\newcommand{\hhT}{\hat A}
\newcommand{\hT}{\widehat T}
\newcommand{\hTi}{\widehat{\mathbf{T}}}
\newcommand{\cF}{\mathcal F}
\newcommand{\tcF}{\widetilde{\mathcal F}}
\newcommand{\bcJ}{\bar\cJ}

\newcommand{\mm}{\mathbb{M}}
\newcommand{\bmm}{{\mathbf m}}
 
\newcommand{\nin}{\not \in}
\newcommand{\cG}{\mathcal G}
\newcommand{\cB}{\mathcal B}
\newcommand{\cQ}{\mathcal Q}
\newcommand{\id}{\mathbb I}

\newcommand{\h}{\hat} 
\newcommand{\atom}{\mathbb G}
\newcommand{\tS}{\tilde S} 

\date{}
\title{Implementing Bell causality  in  Quantum Sequential Growth} 
\author{Ritesh Srivastava, Sumati Surya \\ Raman Research Institute,
  CV Raman Avenue, \\ Sadashivanagar, Bengaluru, 560080} 

\begin{document}
\maketitle
\begin{abstract}
We explore different  implementations of the quantum Bell causality (QBC)  condition 
in  the quantum sequential growth (QSG) dynamics of causal set quantum
gravity, for 
non-commuting transition operators. Assuming a non-singular dynamics we show that for the
two most natural choices of operator orderings for the  QBC, the
transition operator algebra $\cA$ reduces to a 
commutative one. As a third choice, we take the operator ordering to
depend on the size of the precursor set. We find several new
commutation relations which further constrain $\cA$ but 
do not imply commutativity. On the other hand, if any of the
generators of the  ``antichain 
subalgebra''  $\cQ$  belongs to its center, then this implies
commutativity of $\cA$.  The complexity of the algebra prevents us
from obtaining a general form for the transition
operators, which hinders computability. In an attempt to construct 
the simplest non-trivial  $d=2$ representation of $\cA$, we find  that
a Pauli matrix representation of the generators of $\cQ$ leads to
inconsistencies, implying that if a non-trivial representation exists,
it must be  higher dimensional. 
Our work can be viewed as a first step
towards finding a 
non-commutative realisation of QSG. 
\end{abstract}   

\thispagestyle{empty}
 \setcounter{page}{1}
 \section{Introduction}

In the causal set approach to quantum gravity, the  spacetime continuum is replaced by a locally finite partially
  ordered set or causal set, with the continuum emerging as an
  approximation in a 
  very specific, well defined sense \cite{blms,lr,book}.  At the
  coarse-grained level, the dynamics is described by the discrete path
  integral or path-sum, whose measure is determined by the discrete
  Einstein-Hilbert or Benincasa-Dowker-Glaser(BDG)  action
  \cite{bd,dg,glaser}, where the sum is taken over the space of all
  causal sets.  The BDG action has been shown to
  suppress the super-exponential entropic dominance from
 the Kleitmann-Rothschild (KR) and
  other layered orders \cite{kr,dharone,pst} which are
  non-continuumlike \cite{lc,ams,ccsone,ccstwo}. While this does not
  imply the emergence of the continuum, other studies on statistical
  dynamics suggest the existence of a ubiquitous continuum phase
  \cite{2dqg,hh,fss,2dcyl}. 

At its foundation,  however, causal set theory (CST)  replaces   the continuum spacetime geometry
with ``order'' and ``number''. Hence  a natural question is whether the continuum
 inspired dynamics of the path-sum may have more order
 theoretic origins. The Rideout-Sorkin (RS)  sequential growth paradigm 
 offers the possibility of such an alternative, and its
 classical dynamics has been well studied 
 \cite{csgone,csgtwo,observables}.  In this paradigm, a causal
 set is grown element by element, starting with a single element, such
 that the new element at every {\sl stage} $n \in \bN$  lies either to the future of an
 existing element or is unrelated to it\footnote{See
   \cite{bdz} for generalisations to include past and future
   growth dynamics.}.  This ``internal temporality'' ensures that
the dynamics is intrinsically  causal, while the growth to the
 future means that all the causal sets thus generated are {\sl past
   finite}, namely, every element has a finite past.  Classically, the dynamics
 is defined by transition probabilities, giving rise to a stochastic
 triple or measure space $(\Omega, \fA, \mu)$ where $\Omega$ is the
 set of all countable past finite causal sets, $\fA$ is the event algebra
 generated by finite causal sets, and $\mu$ is the classical
 probability measure on $\fA$.  Apart from internal temporality, the
 Rideout-Sorkin   classical sequential growth
 models (CSG) satisfy (i) the Markov sum rule (MSR) for probabilities (ii)
 path independence, or general covariance (GC) and, of most interest to us
 in this work,   (iii) spectator independence, or Bell Causality
 (BC). As shown in \cite{csgone} the resulting dynamics is startlingly
 simple, the most general case being determined by a single parameter
 at every stage, even when some of the probabilities are allowed to
 vanish \cite{rv,dsobs,dzbreaktwo}.

 As generated,  the causal sets in
 $\Omega$ carry labels, with each element being labelled by the stage
 at which it is created. Any covariant or label invariant statement 
 however, must be made on the algebra of unlabelled causal sets. Since
 the classical measure extends from  $\fA$  to its sigma algebra
 $\fS$, one can construct appropriate quotients of $\Omega$ to define
 a covariant subalgebra \cite{observables,dsobs}.  The CSG
 models typically overcome the super-exponential entropy of the 
 KR-type layered  orders and produce more cosmological-like causal sets \cite{csgrg}. 

 As pleasing as the CSG models are, they are  nevertheless classical,
 and must be replaced by more general quantum sequential growth (QSG)  models. Sorkin's quantum measure
 formulation \cite{qmeasureone,qmeasuretwo,dowgha} replaces classical
 probability in a  stochastic 
 triple with a quantum measure, or decoherence functional $D: \fA
 \times \fA \rightarrow \bC$, which is (i) Biadditive (ii) Hermitian
 and (iii) Strongly positive (see Section \ref{qsg.sec} for details). 
 The third condition is not strictly required for a quantum theory, but by including it one
 can replace $D(.,.)$ by a {\sl vector measure} $\ket{.}$ defined on
 a {\sl histories hilbert space} $\cH$ itself constructed from $\fA$ and
 $D(.,.)$ \cite{histories}. We will henceforth work with this latter structure.
In \cite{djs,sz}, complex sequential growth ($\CSG$) models with 
$\cH \simeq \bC$ were constructed,  and it was shown that  
$\ket{.}$ extends to $\fS$  when certain convergence conditions are
met, which in turn gives rise to a class of fully covariant quantum
dynamics. The computability of these models rests on the fact that
these amplitudes   form a commutative algebra. Thus the conditions
(i)-(iii) of CSG, namely MSR, GC
and BC, can used ``as is'', with the dynamics being determined
by a single {\it complex} parameter at every stage. 

Restricting to $\cH \simeq \bC$ however means that $D(.,.)$ 
is limited to the simple product form $ D(\alpha,\beta)=A^*(\alpha)
A(\beta)$, where $A:\fA \rightarrow \bC$.  This is clearly not the
most general form for $D(.,.)$ even for simple quantum systems;
unitary dynamics for a quantum particle for example requires that $ D(\gamma_1,\gamma_2)=A^*(\gamma_1)
A(\gamma_2) \delta (\gamma_1(T)-\gamma_2(T))$, where $T$ is a
truncation time. One imagines that quantum gravity should possess an
even  richer set of possibilities, which requires  $\dim(\cH) >1$.  The $\CSG$ models are a
subset of a broader class of quantum sequential growth (QSG)  models  
in which the transition amplitudes (valued in $\bC$) are replaced by transition
operators. While the MSR and the GC generalise quite easily, there are
operator ordering ambiguities  in the 
implementation of the quantum BC (QBC) condition. 

In this work we explore three possible operator ordering
choices for the QBC while assuming that each of the transition operators 
are non-singular. For our first two, and arguably most natural  choices,
we find that the transition operator algebra $\cA$ reduces to a
commutative one thus leading us back to the $\CSG$ models. Our
construction uses  properties of the ``antichain subalgebra'' $\cQ$, generated from
the transition operators $\hQ_n$ from the $n$-element
antichain $\achain_n$ to the  $(n+1)$-element
antichain $\achain_{n+1}$.  For the
third choice, we use the size of the precursor set of the
new element to determine the operator ordering.  The QBC condition
is therefore not symmetric and an additional relation is required  when the sizes of this set are
equal.  We find several new commutation
relations in $\cQ$, but these are insufficient to prove commutativity.
However, if  any one of the generators of $\cQ$ is assumed to assumed
to belong to the center of $\cQ$, then this implies that  $\cA$ is commutative. Without this
assumption however, the complexity of the algebraic
relations prevent us from finding general form of  the transition
operator.  In an attempt to realise $\cA$ we  construct
a $d=2$ representation of the $\cQ$   in terms of  Pauli
matrices but find that it cannot be realised as a subalgebra of $cA$.
This suggests that any non-trivial representation of $\cA$, if it
exists, must be of higher dimensions.

In Section \ref{prelims.ssec} we introduce the Sequential Growth
Paradigm, along with various definitions that are used later in the
paper. The growth generates a labelled poset $\cP$ of finite element causal
sets, which is a tree and the associated unlabelled poset $\cP'$ which
contains loops.  We then give a kinematic  definition of transitions
that are {\sl Bell partners}, those that are {\sl Bell pairs}, and
{\sl Bell families}  of Bell pairs. The growth gives rise to an  event algebra $\fA$ generated
by  {\sl cylinder
sets} $\cyl{c_n}$ associated with every finite element  labelled causal set $c_n$
in $\cP$.  The set of $\cyl{c_n}$ for a given $n$ provide a disjoint partition of
$\fA$. 

In Section
\ref{sgd.ssec} we introduce the general sequential growth dynamics,
for a {\sl generalised measure} which is finitely additive. We also introduce the 
generalised requirements  of MSR, GC and as well as  BC. Using a
non-disjoint partition of cylinder sets, we find a useful identity for
the generalised measure.  In Section \ref{csg.ssec} we briefly review
the results of the RS CSG models,  highlighting  the key elements which we
will use as templates for constructing QSG models.

We finally turn to the QSG models in Section \ref{qsg.sec}. We begin
with a recap of the quantum vector measure, and in Section
\ref{qsgd.ssec}  we define the QSG models most generally. Section
\ref{qbc.ssec} contains our main results. We begin by considering a very general implementation
of QBC from which we can prove two important Lemmas. We then focus on three implementations: (i) Time
Ordered Bell Causality (TOBC) in Section \ref{tobc.ssec},  (ii) Non-Time
Ordered Bell Causality (NTOBC) in Section \ref{ntobc.ssec} and finally
(iii) Causal Past Ordered Bell Causality (CPOBC) in Section
\ref{cpobc.ssec}. We end with a discussion of our results in Section
\ref{discussion.sec}. Appendix \ref{appendix.app} contains details of the 
calculations using the Pauli matrix algebra.

\section{The Sequential Growth Paradigm}
\label{secone.sec}

\subsection{Preliminaries}
\label{prelims.ssec} 

{A {\sl causal set}   $c$  is a locally 
  finite partially  ordered set,  i.e., a set with the partial order
  $\prec$ which $\forall \, \, e,e',e'' \in c$ satisfies 
\begin{enumerate}
\item Acyclicity : If $e\prec e'\Rightarrow e' \nprec e$  
\item Transitivity:  $e\prec e'$ and $e'\prec e''$ $\Rightarrow e\prec
  e''$,  
\item Local finiteness:  $|\mathbf{I}(e,e')|<\infty$ where
  $\mathbf{I}(e,e')\equiv\{ e'' \in c:  e\prec e'' \prec e'\}$
\end{enumerate}
}

In this work we will concern ourselves with finite $n$-element causal
sets, which we denote by $c_n$. For every $e \in c_n$, we define  
\begin{eqnarray}
\Past(e)&\equiv&\{e' \in c_n: e' \prec e\}\nonumber \\ 
\Fut(e)&\equiv &\{e'\in c_n : e\prec e'\}. \nonumber \\ 
\IPast(e)&\equiv&\{e' \in c_n: e' \preceq  e\}\nonumber \\ 
\IFut(e)&\equiv &\{e'\in c_n : e\preceq e'\}. \nonumber
\end{eqnarray}
An element $e \in c_n$ is said to be {\sl maximal} if
$\Fut(e)=\emptyset$  and it is said to be {\sl minimal}
if $\Past(e)=\emptyset$.

The Rideout Sorkin (RS)  sequential growth process  generates labelled causal sets, starting from a
single element, with the labelling on each element given by the
{\sl stage}  $n \in \bN$ at which it is added, as shown in Figure
\ref{lposcau.fig}  \cite{csgone}.  It obeys the principle of {\sl
  internal temporality} which means that the new element $e_n$ which
has been added at stage $n$, is  either to the future of an existing
element, or is  unrelated to it. The causal set thus generated is (i) {\sl
  naturally labelled} or carries a {\sl linear extension}, namely,
for $e_i, e_j \in c_n$, $e_i \prec e_j \rightarrow i <j$ (ii)  past
finite, i.e., $\forall e_i \in c_n,  \, \, |\Past(e_i)| < \infty$ (see
\cite{bdz} for generalisations).   More generally, an $n$-element
causal set is said to be {\sl labelled}  if each of its
elements is assigned a unique label from $1$ to $n$. The condition of internal temporality
ensures that only natural labelings occur in the growth 
process.  In this work, an  {\sl unlabelled causal set}  corresponds to the equivalence class
of all natural  labelings of the  causal set.

Figure \ref{lposcau.fig} shows  the growth tree, dubbed the {\sl labelled
poscau} $\cP$,      with all the
possible causal sets that can be generated upto stage 
 $n=3$, starting from the single element causal set with element
 $e_1$.  At stage $n=2$, the new element $e_2$ has two ways it can be
 added: either to the future of $e_1$ or unrelated to it. Thus, the
 set $\Omega_2$ of $n=2$-element labelled causal sets in  $\cP$
 contains the totally ordered $2$-element chain $\chain_2$ and the
 totally unordered $2$-element antichain $\achain_2$. 
 Each appears
 exactly once, since $\chain_2$ has a single natural ordering, and
 any relabelling of $\achain_2$ gives back the same labelled causal
 set. At stage $n=3$, we see that the causal set   \pcauset{2,3,1}
 appears with all of its three possible natural labelings. 

Each node of $\cP$ is a finite element naturally labelled causal 
set,  and emerges  from a single {\sl parent} node. We will
adopt the familial terminology of \cite{csgone}, and call the nodes
generated from a (single) parent  $c_n$, its
{\sl children}\footnote{This is also referred to as the index, upward degree or 
  targets of a  node in a directed tree.}. The  number of 
$n+1$-element children $c_{n+1}^j$ of $c_n$ varies from node to node
and hence we  label them over a  $c_n$-dependent  index set
$\cJ_1(c_n)$.  More generally, the index set of the stage $n+m$  {\sl descendents}  of
$c_n$  are denoted by  $\cJ_m(c_n)$.  As we go up stage by
stage in  $\cP$,  we generate the sample spaces  $\Omega_n$  of $n$-element
naturally labelled causal sets and eventually,  the sample space $\Omega$ of countable,
past finite causal sets.

\begin{figure}[H]
\centering
\includegraphics[scale=0.3]{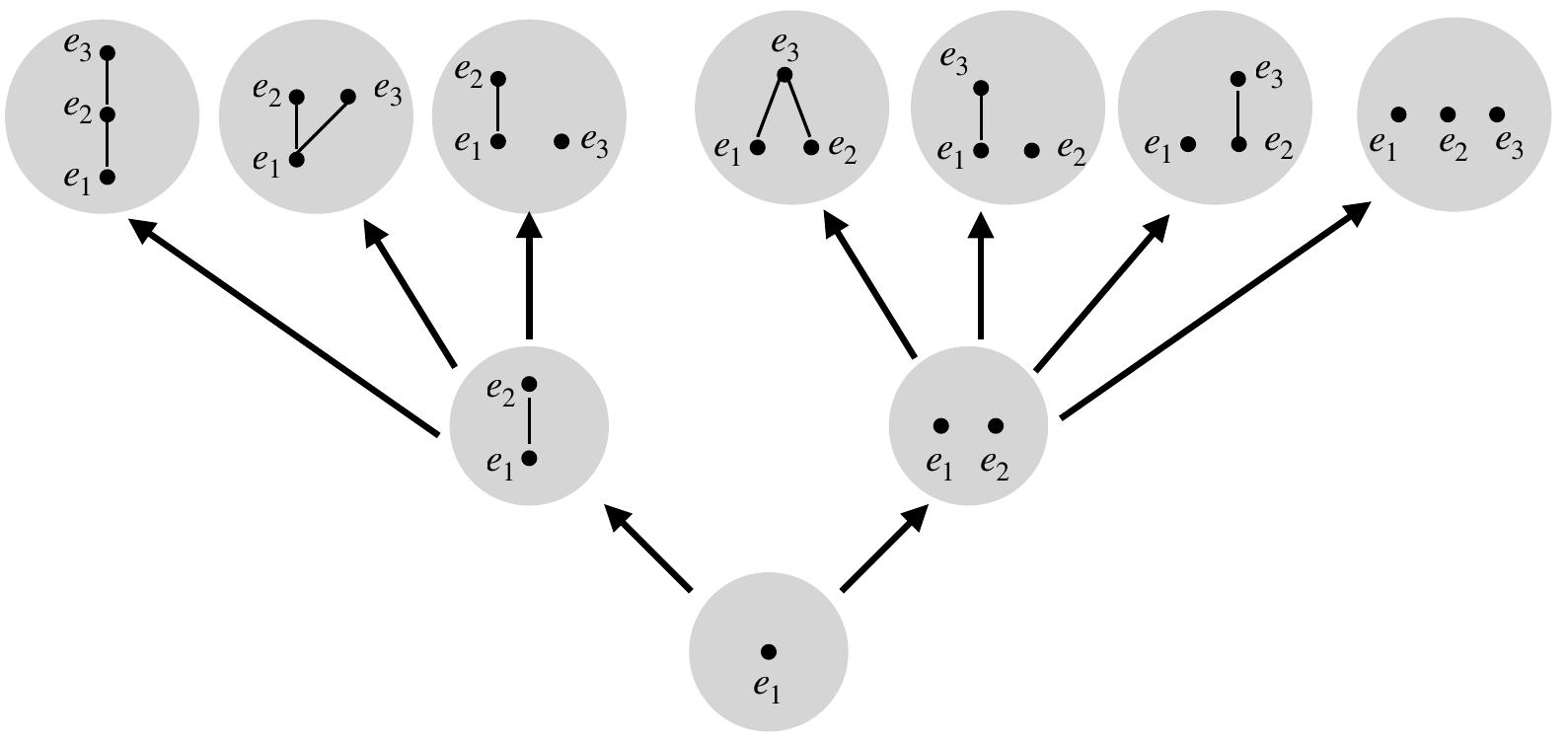} 
  \caption{Labelled Poscau  $\cP$. } 
\label{lposcau.fig} 
\end{figure}

Already at stage $n=3$ of  the growth process, different
natural labelings of \pcauset{2,3,1} manifest themselves.  Any
physical observables we wish to construct however, must be label
independent. Hence,  for example  the causal set
\pcauset{2/label=2,3/label=3,1/label=1}, which appears in $\cP$ 
belongs to the  equivalence class $\pcauset{2,3,1} = [\pcauset{2/label=2,3/label=3,1/label=1}]$ of all natural
relabelings, where the equivalence relation is
$\pcauset{3/label=3,1/label=1,2/label=2} \sim_L
\pcauset{2/label={2},3/label=3,1/label=1}
\sim_L\pcauset{2/label={1},3/label=3,1/label=2}$.

This equivalence  gives rise to the poset of causal sets, {\sl poscau}
$\tcP$ of  unlabelled causal
sets, whose nodes are unlabelled causal sets $\tcn$\footnote{ We will
  soon drop the tildes $\tilde{X}$ for ease of notation. See the
  note at the end of this Section and just before the start of Section \ref{sgd.ssec}.}
Fig \ref{poscau.fig}  shows $\tcP$ upto stage $n=3$, which can be
contrasted with Fig \ref{lposcau.fig}.   Though itself a poset, $\tcP$
is {\it not}  a tree since the
same unlabelled causal set can be reached in multiple unlabelled
ways. 
\begin{figure}[H]
\centering
\includegraphics[scale=0.3]{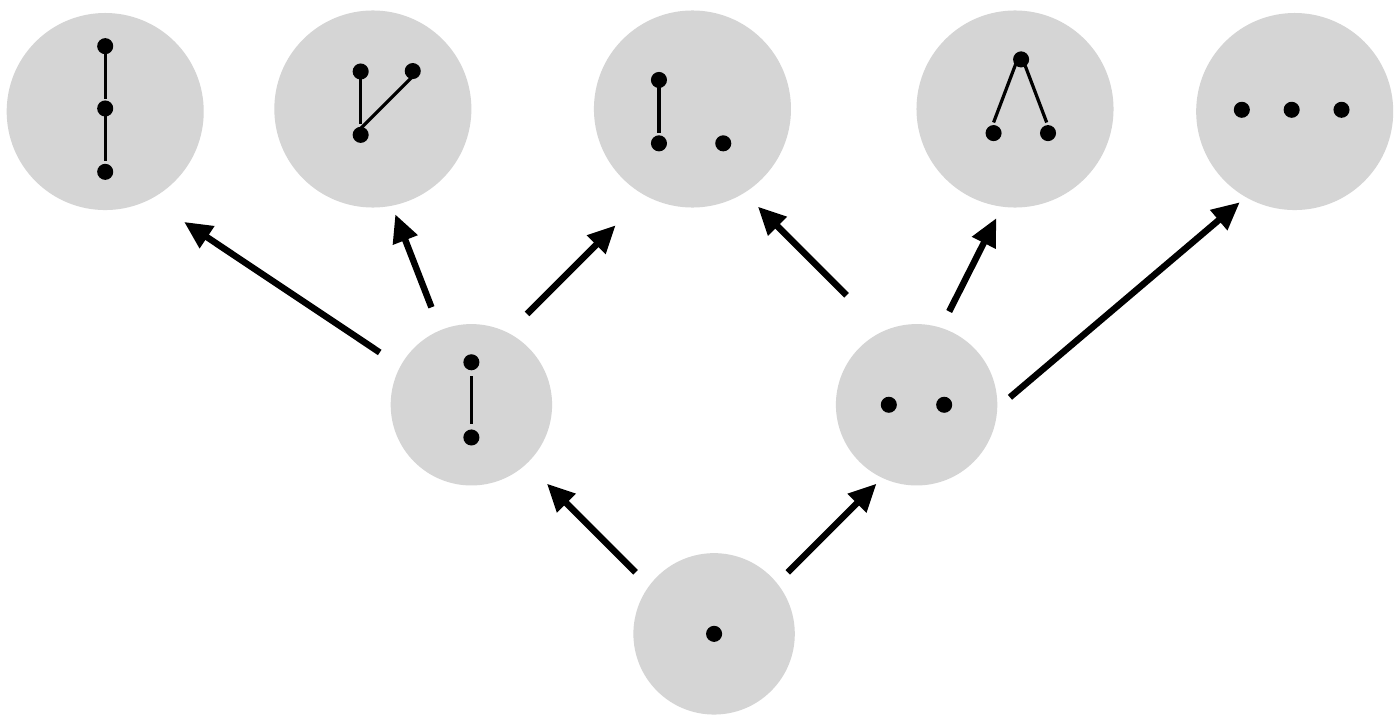} 
  \caption{Poscau  $\tcP$. } 
\label{poscau.fig} 
\end{figure}

We define a {\sl transition} in $\cP$ as a map
$\mathcal{T}_n^{j}:c_n\rt c_{n+1}^{j}$ where $j \in \cJ(c_n)$, which
indexes the children of $c_n$.  This transition adds in the new
element $e_{n+1}$ to $c_n$.  Hence transitions in $\cP$ are 
also partially ordered as shown in Fig \ref{OT.fig}.
\begin{figure}[H]
\centering
  \begin{tikzpicture}[-stealth]
  \def\ystep{1.5cm}
  \def\xstep{0.5cm}
  \begin{scope}[nodes={draw, thin, circle, minimum size=0.5cm}]
  \node (C0) at (-4.5,0) {\pcauset{3,2,1}};
 \node (C1) at ( -1.5, 0) {\pcauset{3,2,4,1}};
\node (C2) at (1.5,0){\pcauset{3,2,4,5,1}};
\node (C3) at (4.5,0){\pcauset{3,2,4,6,5,1}};
\end{scope}
\begin{scope}[prob arrow, font={\small}]
  \drawprobarrow{C1}{\mathcal{T}_4}{C2}
  \drawprobarrow{C0}{\mathcal{T}_3}{C1}
  \drawprobarrow{C2}{\mathcal{T}_5}{C3}
  \end{scope}
 \end{tikzpicture}
\caption{Ordering of transitions} \label{OT.fig} 
 \end{figure}
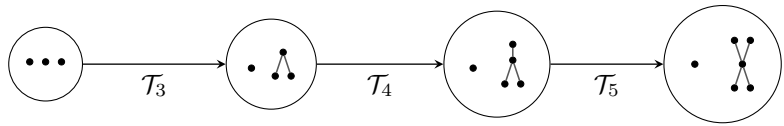

As in \cite{csgone}  it is useful to kinematically deconstruct a 
transition $\cT_n: c_n \rt  c_{n+1}$ using the ordering relation. 
For the new  element 
$e_{n+1} \in c_{n+1}$ We define the {\sl precursor} set $\pre(c_{n+1})$  as  $\pre(c_{n+1})\equiv
\Past(e_{n+1})$ while   its  {\sl spectator} set $\spec(c_{n+1})
\equiv c_n\setminus \Past(e_{n+1})$. The latter are the elements in $c_n$
that do not participate in the transition kinematically.   Thus, in
the transition
\begin{equation}
  \cT_3: \pcauset[labeled]{3,2,1} \rightarrow
  \pcauset{3/label=3,2/label=2,4/label=4,1/label=1}
\end{equation}
the elements $e_2,e_3$ are in the precursor set and $e_1$ is in the
spectator set.

A transition $\cG_n$ 
is moreover said to be
{\sl gregarious}   if  $\Past(e_{n+1})=\emptyset$, while in a
{\sl timid} transition $\cT_n^T$,  $\Past(e_{n+1})= c_n$. The child associated
with the timid transition will be denoted $c_{n+1}^T$ and that with a
gregarious transition $c_{n+1}^g$.

For the pair
of transitions,  $\cT_n: c_n \rt  c_{n+1}$ and  $\cT_n': c_n \rt
c_{n+1}'$, we define
the {\sl common spectator set} $\cspec(c_{n+1},c_{n+1}') \equiv \spec(c_{n+1}) \cap
\spec(c_{n+1}')$.    
We define a {\sl Bell partner}  of a 
transition $\cT_n:c_n \rightarrow c_{n+1}$ as $\cT_m:c_m \rightarrow
c_{m+1}$, where $ c_m$ is obtained from $c_n$ by  adding or removing
common spectators to and from $\cspec(c_{n+1})$. Since the precursor
set of the new elements in both transitions is identical, they are
kinematically  essentially the same transition. 
  
In order to formalise this we define a  {\sl Bell pair}  of
transitions at stage $n$
  $\cT_n:c_n \rightarrow c_{n+1}$,  $\cT_n':c_n \rightarrow c_{n+1}'$,
 such that $\cspec(c_{n+1},c_{n+1}') \neq \emptyset$, with the related 
 Bell pair 
$\cT_m:c_m \rightarrow c_{m+1}$,  $\cT_m':c_m
\rightarrow c_{m+1}'$ where $c_m$ is obtained from either 
adding or removing common spectators  in $\cspec(c_{n+1},c_{n+1}')$ to
or from $c_n$.   
  
We will define a {\sl Bell family} $\{(\cT_n,\cT_n') \}$ to be the set
of all related Bell pairs  with $n=[n_0, \infty)$  where $n_0$ is
defined by $\cspec(c_{n_0+1},c_{n_0+1}') =\emptyset$.  If one started off  with 
  $n_0$, then one can only add common spectators, to get another
  related $n$-element Bell pair i.e., $n > n_0$.  
Note that every non-gregarious transition can be Bell paired with a
gregarious transition. 
  
Figure \ref{bellpair.fig}   shows the two related Bell pairs
$(\cT_3,\cT_3')$ and $ (\cT_5,\cT_5')$, where the transitions
$\cT_3$ and $\cT_5$  have the same precursor
  set  as do the transitions 
  $\cT'_3$ and $\cT'_5$. Moreover, the pair $(\cT_5,\cT_5')$ has  a 
 non-empty  common spectator set, while that for  $(\cT_3,\cT_3')$ is
 empty, and hence $n_0=3$.  
\begin{figure} 
\begin{center}
\begin{tikzpicture}[-stealth]
  \def\ystep{2cm}
  \def\xstep{0.5cm}
  \begin{scope}[nodes={draw, thin, circle, minimum size=0.5cm}]
 \node (C1) at ( -5, 0) {\pcauset{3,2,1}};
 \node (C12) at (-6.5, 1*\ystep) {\pcauset{3,2,1,4}};
 \node (C21) at ( -3.5, 1*\ystep) {\pcauset{3,4,2,1}};
\node (C2) at (4,0) {\pcauset{5,4,2,1,3}};
 \node (C642153) at (2 , 1*\ystep) {\pcauset{6,4,2,1,5,3}};
 \node (C4321) at (6 , 1*\ystep) {\pcauset{6,4,5,2,1,3}};
\end{scope}
\begin{scope}[prob arrow, font={\small}]
  \drawprobarrow{C1}{\mathcal{T}_3}{C12}
\drawprobarrow{C1}{\mathcal{T}'_3}{C21}
\drawprobarrow{C2}{\mathcal{T}_5}{C642153}
\drawprobarrow{C2}{\mathcal{T'}_5}{C4321}
  \end{scope}
\end{tikzpicture}
\caption{Bell related pairs of transitions}
\label{bellpair.fig}
\end{center}
\end{figure}
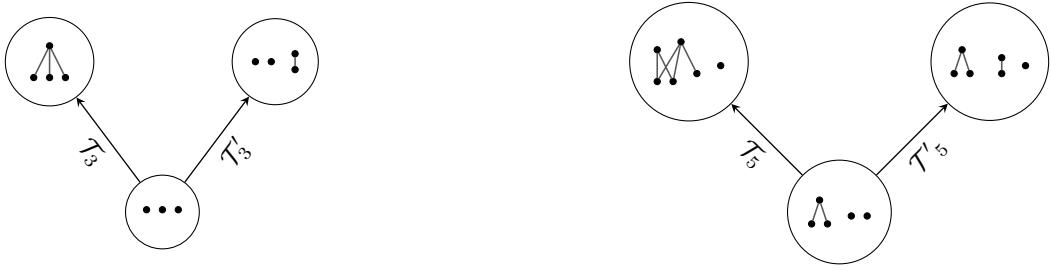

While each node in $\cP$ is a finite element causal set, it also
represents a {\sl cylinder subset} $\cyl{c_n}$ in $\Omega$ containing
all countable labelled causal sets  whose  first $n$-elements are
$c_n$. This is similar to a finite time $T$ path $\gamma_T$in the discrete random
walk representing a subset of infinite time walks whose first $T$
steps are $\gamma_T$.  Because of the tree structure of $\cP$, the
$\cyl{c_n}$ possess a  nesting structure: if  $\cyl{c_n} \cap
\cyl{c_m}  \neq \emptyset$ and wlog $n\geq m$, then $\cyl{c_n}\subset
\cyl{c_m}$.   Consider the set $\{c_n^i\}$ of all the
$n$-element nodes in $\cP$, where $i \in \cI(n) \equiv \{1, 2, \ldots, |\Omega_n| \}
$.   The set $\cZ_n$ of $n$-element cylinder sets $\{\cyl{c_n^i} \}$ at
stage $n$ thus provides a disjoint  partition of $\Omega$. For each $i
\in \cI(n)$ we have the set theoretic condition,
\begin{equation}
\cyl{c_n^i} = \bigsqcup_{j \in \cJ(c_n^i)} \cyl{c_{n+1}^j}, \label{cylpart.eq} 
  \end{equation} 
i.e., each cylinder set at stage $n$ can be expressed as a disjoint
union of the cylinder sets formed from its children $\cyl{c_{n+1}^j}
\in \cZ_{n+1}$. 
Each in turn can be partitioned into cylinder sets in $\cZ_{n+2}$, and
so on, and hence in general, for any $m$, 
\begin{equation}
\cyl{c_n^i} = \bigsqcup_{k\in \cK(c_n^i,m)}
\cyl{c_{n+m}^k}, \label{gencylpart.eq} 
  \end{equation} 
where $\cK(c_n^i,m)$ is the index set associated with the $n+m$
descendents of $c_n^i$. Let $\cZ$ denote the set of all finitely
generated cylinder sets, i.e., $\cZ \equiv \bigcup_{n \in \bN} \cZ_n$. 
We note that not only can $\cyl{c_n}$ be
partitioned into a disjoint union of the $\cyl{c_{n+1}^i}$ for $i \in
\cJ(c_n)$  but also that at every stage $n$,    $\Omega$ 
disjoint cylinder-set partition, 
\begin{equation}
\Omega=\bigsqcup_{z \in \cZ_n} z. 
\end{equation}

\begin{definition}
  The event
algebra  $\fA$ over $\Omega$ is generated from $\cZ$ by finite set operations.  
\end{definition}

We now introduce a non-disjoint partitions of $\cyl{c_n}$
which will be useful in our constructions.
Let $\mm(c_n)=\{e_{m_1}, e_{m_2}, \ldots, e_{m_\bmm} \}$ denote the set of maximal elements in $c_n$. Hence
$\bmm \leq n $, with the equality holding iff
$c_n=\achain_n$. Let $e_{n+1}^i$ denote the new element added to $c_n$
to get $c_{n+1}^i$, $i \in \cJ(c_n)$. Then for any $e_k \in c_n, e_k
\nin \pre(c_{n+1}^i)$, $\IFut(e_k) \subseteq \spec(c_{n+1}^i) \subset
c_n$, and hence $\spec(c_{n+1}^i)=\cup_k \IFut(e_k) $, over all such
$e_k$ is a future set.  If $\spec(c_{n+1}^i) \neq \emptyset$ then there exists an $e_{m_k}
\in \mm(c_n)$ such that $e_{m_k} \in \spec(c_{n+1}^i)$.

Let $\setzeta_k$ be the set of
children of $c_n$ for which the 
maximal element $e_{m_k}$ is in the spectator set  (but possibly not the only
one) 
\begin{equation} 
\setzeta_k \equiv \{\cyl{c_{n+1}^{j(k)}} | j(k) \in \cJ(c_n), e_{m_k}
\nprec e^{j(k)}_{n+1} \},  e_{m_k} \in \mm(c_n). 
\end{equation}
Define the sets 
\begin{equation}
\zeta_k \equiv \bigsqcup_{\setzeta_k} \cyl{c_{n+1}^{j(k)}} \in \fA,
\quad k \in {1, \ldots,  \bmm} 
\end{equation}  
The  $\zeta_k$ are themselves not disjoint and none contain
$\cyl{c_{n+1}^T}$, where $c_{n+1}^T$ is the timid child of $c_n$.
 Moreover, for every  non-timid
transition from $c_n$ there exists at least one $k$ such that $e_{m_k}
\in \mm(c_n)$ is a spectator element. Thus  
\begin{equation}
  \cyl{c_n^j}\setminus\cyl{c_{n+1}^T} = \bigcup_{k=1}^{\bmm}
    \zeta_k \quad \Rightarrow \quad 
  \cyl{c_n^j} = \bigcup_{k=1}^{\bmm} \zeta_k \, \, \bigsqcup
  \cyl{c_{n+1}^T}. \label{zetal.eq} 
\end{equation}

We will find this alternative partitioning of $\cyl{c_n}$ to be
useful later on in the paper when we  find the general form of
transition amplitudes.

\noindent {\bf Note:} In order to simplify notation, we will
henceforth 
not distinguish between labelled and unlabelled causal sets and hence
drop the distinguishing tilde notation $\tilde{X}$ that we have placed over the
unlabelled causal sets and poscau.   The choices 
should be clear from the context. Indeed, for most
of what follows, we will use unlabelled poscau (without using the tilde
notation  $\tilde{X}$), keeping in mind that
one must take into account the 
different possible natural labelings giving rise to distinct labelled
children from a given parent causal set. 

\subsection{Sequential growth dynamics(SGD)}
\label{sgd.ssec}

The sample space $\Omega$ and the event algebra $\fA$ form the first
two elements of the classical and quantum measure triples we will
construct, for a ``generalised measure'' which we define below. 

\begin{definition} \label{gmu.def} 
Let $\Omega$ be a sample space with event algebra $\fA$.   We define a
(finitely additive) 
{\sl  generalised measure}  as a map $\gmu: \fA \rightarrow
\cR$, where $\cR$ is an abelian group  with
identity $0$ such that $\gmu(\sqcup_i\alpha_i) = \sum_i \gmu(\alpha_i) $ for any finite
disjoint set $\{ \alpha_i \in \fA\}$ and (ii)
$\gmu(\emptyset)=\id$. 
\end{definition}

 A sequential growth dynamics is an assignment of a generalised
 measure $\gmu$ on $\fA$,  such that every transition $\cT_n: c_n
 \rightarrow c_{n+1}$ is associated with a
 {\sl transition  map}  $T_n: \cR \rightarrow \cR$ such that 
 \begin{equation}
\gmu(\cyl{c_{n+1}}) = T_n(c_{n+1}) \circ \gmu(\cyl{c_n}) \Rightarrow
\gmu(\cyl{c_{n+1}}) = T(\gamma(c_{n+1})) \circ
\gmu(\Omega) \label{gammacn.eq} 
  \end{equation}
  where $\gamma_{n+1}$ denotes the {\sl path}  from $c_1$ to
  $c_{n+1}^i$ in $\cP$ and
  \begin{equation} T(\gamma(c_{n+1})) \equiv T_n(c_{n+1})
    \circ T_{n-1}(c_{n}) \ldots \circ T_1(c_{2}) \circ T_0(c_1) \end{equation}
  denotes the composition 
  transitions from $c_1$ along a path $\gamma(c_{n+1})$ in $\cP$. 
Eqn. \eqref{gammacn.eq} implies that 
  $\gmu(\cyl{c_n})$ can be obtained purely from the
  transition maps associated with  $\gamma(c_n)$, given
  $\gmu(\Omega)$. 
For CSG,  $\cR = [0,1]$ and the $T_n^i \in [0,1]$  are probabilities, 
as are the $\gmu$. 
  
Before going into specifics, we notice that the  
{\it kinematic structure} of encoded in  $\cP$, already gives us
strong constraints on a covariant and causal sequential growth
dynamics whether it is classical or quantum, as long as the
generalised measure satisfies some basic constraints.

For $\gmu$ a  generalised finitely additive measure  $\gmu: c_n \rt
\cR$ (Def. \ref{gmu.def}) we note that the following general
features follow from the kinematical structure of $\cP$ with minimal
physical assumptions as stated:  

\begin{enumerate}  
\item  Since $\cyl{c_n}= \bigsqcup_j
  \cyl{c_{n+1}^j}$ where $j \in \cJ(c_n)$,
  \begin{equation} \gmu(\cyl{c_n})=\sum_{j
      \in \cJ(c_n)} \gmu(\cyl{c_{n+1}^j}). \label{gmumsr.eq}
  \end{equation}   We refer to this as the
  \textbf{Markov Sum Rule(MSR)}. Note that this requires no further
  assumptions on $\gmu$.  

\item  Assuming further that   $\gmu$ is label
  independent or {\sl  covariant},  $\gmu(\cyl{c_n})$ and $\gmu(\cyl{c_n'})$
  should lead to the same physical observable whenever $c_n \sim_L c_n'$.   This is the condition of
\textbf{General Covariance (GC)}  and its specific implementation will
depend on $\gmu$. For example in CSG, we require  the equality of 
probabilities $P(\cyl{c_n})=P(\cyl{c_n'})$.  

\item Assuming that we also require that $\gmu$ is intrinsically
  causal and hence does not essentially
  depend on the spectator sets but only on the precursor sets (in the
  sense of the transition maps), i.e., it satisfies a form of
  \textbf{Bell Causality (BC)},  then $\gmu(c_n)$ and $\gmu(c_m)$ cannot be independently
  determined. 
\end{enumerate}

Apart from the MSR which only requires finite additivity of $\gmu(.)$,
the remaining conditions are natural from the point of view of CST
since they are required by covariance and causality. Dropping any of
these very broad requirements would mean violating one of these
principles, at least in this sequential sense. In \cite{bm} it is
suggested that covariant and causal measures on $\fA$  may be
recovered without imposing these conditions at the sequential
level, but we will not consider this approach here.

We have used the finite additivity of $\gmu$ for the disjoint
partition of a cylinder sets into its cylinder sets of its children,
Eqn. \eqref{gmumsr.eq}, but we can also
apply it to the non-disjoint partition $\{ \zeta_k \}$, $k  \in \{ 
1, \ldots, \bmm \}$. For every pair $\zeta_{k_1}, \zeta_{k_2}$
\begin{equation}
  \zeta_{k_1} \bigcup  \zeta_{k_2} =  \bigl(\zeta_{k_1} \setminus
  \zeta_{k_2} \bigr)
\bigsqcup \bigl(\zeta_{k_2} \setminus \zeta_{k_1} \bigr) \bigsqcup \bigl(\zeta_{k_1}
\bigcap \zeta_{k_2}\bigr)
\end{equation} 
Thus
\begin{eqnarray} 
\gmu( \zeta_{k_1} \bigcup  \zeta_{k_2}) & = &  \gmu\bigl(\zeta_{k_1} \setminus
  \zeta_{k_2} \bigr) + \gmu\bigl(\zeta_{k_2} \setminus \zeta_{k_1}
  \bigr) + \gmu\bigl(\zeta_{k_1}
                                              \bigcap \zeta_{k_2}\bigr)  \nonumber \\
  &=&  \gmu\bigl(\zeta_{k_1}\bigr)+
\gmu\bigl(\zeta_{k_2}\bigr) -\gmu\bigl(\zeta_{k_1} 
\bigcap \zeta_{k_2}\bigr)
\end{eqnarray} 
It is easy to see that this generalises to
\begin{equation}
\gmu( \bigcup_{k=1}^{\bmm} \zeta_k) = \sum_{k=1}^{\bmm}
\gmu\bigl (\zeta_k\bigr) - \sum_{k_1<k_2}^{\bmm}\gmu\bigl(\zeta_{k_1}  \bigcap
\zeta_{k_2}  \bigr) + \cdots + (-1)^{\bmm-1}
\gmu\bigl(\bigcap_{k=1}^{\bmm} \zeta_k\bigr) \label{gmuzeta.eq}.
\end{equation}
Next, note that the set $\zeta_{k_1} \cap \ldots \cap \zeta_{k_\ell}$ can
also be expressed as 
the disjoint union of  cylinder sets $\cyl{c_{n+1}^{j(k_1, \ldots k_\ell)}}$
which contain  $\{ e_{m_{k_1}}, \ldots,  e_{m_{k_\ell}}\}$ in their spectator
sets, i.e.,
\begin{equation}
\zeta_{k_1} \cap \ldots \cap \zeta_{k_\ell}=\bigsqcup_j
\cyl{c_{n+1}^{j(k_1, \ldots k_\ell)}}
  \end{equation} 
  and hence
  \begin{equation}
    \gmu(\zeta_{k_1} \cap \ldots \cap \zeta_{k_\ell})=\sum_j 
\gmu(\cyl{c_{n+1}^{j(k_1, \ldots k_\ell)}}). \label{gmuintzeta.eq} 
\end{equation}
Thus, Eqn. \eqref{gmuzeta.eq} can be expressed entirely in terms of the
generalised measure $\gmu$ of cylinder sets. 

Since  $\bigcup_{k=1}^{\bmm} \zeta_k$ is the union of cylinder
sets of  non-timid children, using MSR 
\begin{equation}
\gmu(\cyl{c_{n+1}^T})= \gmu(\cyl{c_{n}}) -\gmu(
\bigcup_{k=1}^{\bmm} \zeta_k) .  \label{gmutimid.eq} 
\end{equation}

\subsection{Atomisation and Decimation} 

In CSG \cite{csgone} a process of ``atomisation'' of an $n$-element   causal set
into an $n$-element antichain is used to prove a crucial Lemma. We will
use both this, as well as the process of ``decimation'' to help us
simplify the transition operator algebras.  Fig \ref{AD.fig} shows
the   non-dynamical process of atomisation  of a
causal set  $c_n^{(0)}$ into the antichain $\achain_n$, and its 
associated ``decimation'' into $\achain_k$, where $k$ is the number of
minimal elements in $c_n^{(0)}$.

\begin{figure}
  \centering{
\includegraphics[scale=0.51]{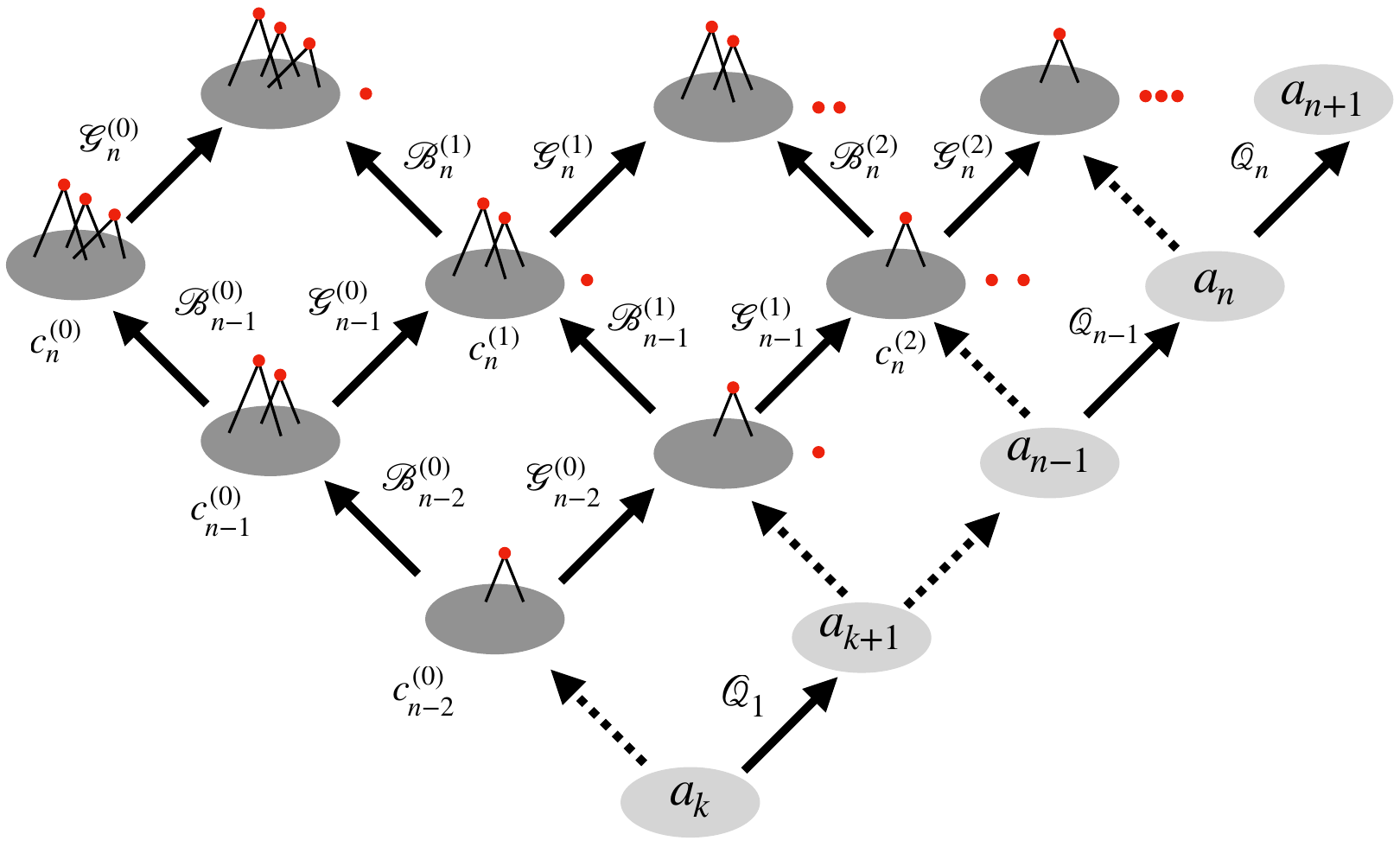}
  \caption{Atomisation takes $c_n$ to $\achain_n$, along an
    atomisation  path  $S(c_n)$ while 
    decimation takes $c_n$ back to $\achain_k$, along the related
    decimation path
    $\Gamma(c_n)$ where $k$ is the
    number of minimal elements in $\achain_n$. While  $S(c_n)$ is not unique for a given $c_n$,
    it determines $\Gamma(c_n)$.} \label{AD.fig} }
  \end{figure}

For any $c_n^{(0)} \nsim \achain_n$, pick a non-gregarious maximal element
$e_{g_1} \in \mm(c_n^{(0)})$, and define $c_n^{(1)} \equiv c_n^{(0)}\backslash{e_{g_1}}
\cup_g e_{g_1}$, where $\cup_g$ denotes a ``gregarious union'', i.e.,
$e_{g_1}$ is a gregarious element in $c_n^{(1)}$. Thus, one takes
out a non-gregarious maximal  element in $c_n^{(0)}$ and adds it in as
a gregarious element to construct $c_n^{(1)}$. Next, if $c_n^{(1)}
\nsim \achain_n$,  we recursively define $c_n^{(2)} \equiv c_n^{(1)}\backslash{e_{g_2}}
\cup_g e_{g_2}$, where now $e_{g_2}$ is a non-gregarious maximal
element in $c_n^{(1)}$. Thus, both $e_{g_1}$ and $e_{g_2}$  are
gregarious elements  in $c_n^{(2)}$. We proceed to construct
$c_n^{(i)}$ recursively, until some finite $k(c_n^{(0)})$ for which  $c_n^{(k)}
\simeq \achain_n$.  Clearly, this {\sl atomisation path}, $S^{(\alpha)}(c_n^{(0)})
\equiv \{ c_n^{(0)}, c_n^{(1)} \ldots  \achain_n\}$ is not
unique. 

Fig \ref{AD.fig} also shows the various transition operators. The gregarious transition $\cG_n^{(0)}: c_n^{(0)} \rightarrow
{c_{n+1}^{(g)}}$,  and the associated non-timid transition operator  $\cB_n^{(1)}: {c_n^{(1)}} \rightarrow
{c_{n+1}^{(0,g)}} $, which adds in the non-gregarious maximal
element $e_{g_1}$ which was removed from  $c_n^{(0)}$ as the  new element in the labelled
version of 
$c_{n+1}^{(0,g)}$.    In general $\cB_n^{(i)}: {c_n^{(i)}} \rightarrow
{c_{n+1}^{(i-1,g)}} $, which adds in the non-gregarious maximal
element $e_{g_i}$ which was removed from  $c_n^{(i-1)}$ as the  new element in the labelled
version of  $c_{n+1}^{(i-1,g)}$. 

We now construct the  {\it decimation} of $c_n^{(0)}$ using the
atomisation path
$S^{(\alpha)}(c_n^{(0)})$, i.e, we construct a related dynamical
{\sl decimation path}  $\Gamma^{(\alpha)}(c_n^{(0)})$ in $\cP'$ from $c_n^{(0}$ back to
the single element causal set $\achain_1$.    Denote  the set of elements that
lead to the atomisation of 
$c_n^{(0)}$ via the path $\Gamma^{(\alpha)}(c_n^{(0)})$ by
${\atom}^{(\alpha)}(c_n^{(0)}) \equiv \{e_{g_1}, \ldots
e_{g_k}\}$. Define $c_{n-1}^{(0)}\equiv c_n^{(0)}\backslash
e_{g_1}$.  It is clear that the gregarious transition from
$c_{n-1}^{(0)}$ results in $c_n^{(1)}$ with $e_{g_1}$ as the
gregarious element. Next, define $c_{n-2}^{(0)}\equiv c_{n-1}^{(0)}\backslash
e_{g_2}$.    Removing the first  $i$  elements in $\atom^{(\alpha)}(c_n^{(0)})$ sequentially from
$c_n^{(0)}$, thus gives us the $(n-i)$-element causal set $c_{n-i}^{(0)} \equiv
c_n^{(i)}\backslash\{e_{g_1}, \ldots,  e_{g_i} \}$.  For $i=k$, $c_{n-k}^{(0)} \equiv
c_n^{(k)}\backslash\{e_{g_1}, \ldots,  e_{g_k}
\}=\achain_{n-k}$. Finally, we can remove the elements one by one from $\achain_{n-k}$ to get
the single element causal set.  Thus we have a path
$\gamma^{(\alpha)}(c_n^{(0)}) \equiv \{ \achain_1, \achain_2, \ldots
\achain_{n-k}, c_{n-k+1}^{(0)}, \ldots c_{n-1}^{(0)}, c_{n}^{(0)}\}$
in $\cP'$. 

Let  $\cB_{n-i}^{(0)}:{c_{n-i}^{(0)}}\rightarrow {c_{n-i+1}^{(0)}}$,
    and let $\cG_{n-i}^{(0)}$ be the gregarious transition from 
    ${c_{n-i}^{(0)}}$.    Comparing with the transition 
    $\cB_n^{(i)}$, for which $\{e_{g_1}, \ldots
e_{g_i}\} \subset \spec(\cB_n^{(i)})$, we find the Bell
related pairs $\{\cB_n^{(i)}, \cG_n^{(i)}\}$ and $\{\cB_{n-i}^{(0)},
\cG_{n-i}^{(0)}\}$.  

To complete this picture, we need another sequence of  causal sets
$c_{n-1}^{(i)}$, constructed similarly to the  $c_{n}^{(i)}$ by
removing maximal elements from $c_{n-1}^{(0)}$ sequentially as
follows. Since  $c_{n}^{(1)} = c_n^{(0)}\backslash e_{g_1} \cup_g
e_{g_1}$ and  $c_{n-1}^{(0)}=c_n^{(0)}\backslash e_{g_1}$, they are
related by the gregarious transition $\cG_{n-1}^{(0)}:
{c_{n-1}^{(0)}} \rightarrow {c_{n}^{(1)}} $. Now define
$c_{n-1}^{(1)}={c_{n-1}^{(1)}} \backslash e_{g_2} \cup_g
e_{g_2}$. The transition $\cB_{n-1}^{(1)}: {c_{n-1}^{(1)}} \rightarrow
{c_{n}^{(1)}}$ thus adds in the element $e_{g_2}$. Generalising,
we can sequentially  define $c_{n-1}^{(i)}={c_{n-1}^{(i-1)}} \backslash e_{g_{i+1}} \cup_g
e_{g_{i+1}i}$ for all $i=1, \ldots k-1$, since $c_{n-1}^{(k-1)}=\achain_{n-1}$.
Similarly, $c_{n-j}^{(i)}={c_{n-j}^{(i-1)}} \backslash e_{g_{i+j}} \cup_g
e_{g_{i+j}}$ for all $i=1, \ldots k-j$, with
$c_{n-j}^{(k-j)}=\achain_{n-j}$.

\subsection{Classical Sequential Growth (CSG) Models}
\label{csg.ssec}

The CSG dynamics of Rideout-Sorkin is based on the measure triple
$(\Omega, \fA, \mu)$ where  $\mu: \fA \rt [0,1]$ is a probability
measure over $\fA$. $\mu(\cyl{c_n})$ is therefore a product of the 
{\sl transition probabilities} $A_k$ associated with the transition
$\cT_k$, with $k$ running over all the transitions from $c_1=a_1=\pcauset{1}$ to $c_n$.

In the CSG models, the dynamics is  required to satisfy
classical/probabilistic versions of the MSR, GC and BC, implemented as
follows, assuming that all the transition probabilities are
non-vanishing\footnote{This condition can be weakened leading to a
  resetting of the dynamics  \cite{rv,dsobs}.}.

\begin{enumerate}
\item {\bf  MSR:}  This follows from finite additivity of $\mu$.   
\item {\bf GC:} $\mu(c_n)=\mu(c_n')$ if  $c_n \sim c_n'$.  
  \item {\bf BC:} For the Bell related pairs
    $\{\{\cT_n,\cT_n'\}, \{ \cT_m,\cT_m'\}\}$,
   \begin{equation} \frac{A_n}{A_n'} \equiv
     \frac{A_m}{A_m'} .  \label{CBC.eq}
   \end{equation}
   \label{csgbc.eq} 
  \end{enumerate} 

As shown in \cite{csgone} these three conditions suffice to rewrite
all transition probabilities at stage $n$ in terms of the transition
probabilities $Q_k$ from $\achain_k$ to $\achain_{k+1}$ for $k \leq
n$, where  $Q_0=1$.  Because we will try to model our arguments for
QSG  around the classical case, we point to two crucial results of
CSG.
\begin{lemma}[Lemma 1 of \cite{csgone}]
Every transition probability $A_n^i:c_n \rightarrow c_{n+1}^i$ at stage $n$ can be expressed in
terms of the gregarious transition probability $G_n:c_n \rightarrow
c_{n+1}^g$ as well as lower stage probabilities.  Here $i \in
\cJ(c_n)$ the index set of children of $c_n$.  \label{csgonelemmaone.lemma} 
\end{lemma} 
\begin{proof}
Every  non-timid transition probability $A_n^i$ has a Bell partner $A_m^i$ for
some $m<n$. Pairing these with their respective  gregarious transition
probabilities 
\begin{equation}
A_n^i = A_m^i G_n G_m^{-1}   
\end{equation} 
Additionally, using the MSR, the timid transition probability is
\begin{equation}
T_n= 1-\sum_{i \in \bcJ(c_n)}  A_m^i G_nG_m^{-1}   
\end{equation}
where $\bcJ(c_n)$ denotes the index set of non-timid children of
$c_n$. 
\end{proof}
\begin{figure}[H]
\centering
\includegraphics[scale=0.55]{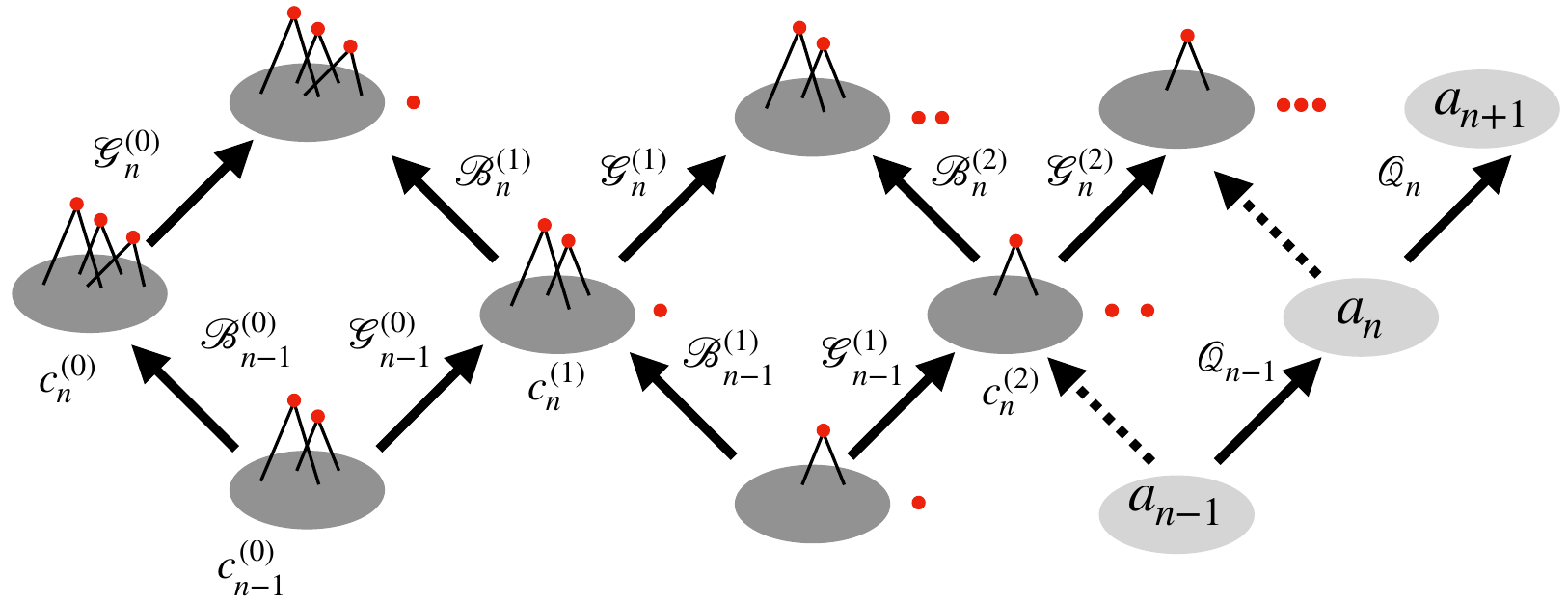} 
  \caption{Atomisation of Gregarious Transitions in
  unlabelled poscau $\cP'$. } 
\label{atomisation.fig} 
\end{figure}
\begin{lemma}[Lemma 2 of \cite{csgone}]
Every gregarious transition probability $G_n$ at stage $n$ is equal to
$Q_n$. \label{csgonelemma.lemma} 
  \end{lemma}
\begin{proof} Consider the transitions in Fig \ref{atomisation.fig}  in
  $\tcP$.  Let $B_k^{(i)}$ be the transition probability associated with the
  transition $\cB_k^{(i)}$, and $G_k^{(i)}$ the transition probability
  corresponding to the gregarious transition $\cG_k^{(i)}$

  From GC,
  \begin{equation}
G_n^{(0)} B_{n-1}^{(0)} = B_{n}^{(0)}  G_{n-1}^{(0)}
\end{equation} 
 while from  BC  
\begin{equation}
\frac{B_n^{(1)}}{G_n^{(1)}}=\frac{B_{n-1}^{(0)}}{G_{n-1}^{(0)}} 
\end{equation}
From these equations, we see that $G_{n}^{(1)}=G_{n}^{(0)}$.
By ``atomising'' in the manner shown  we can now extend the figure to the right, all the way to the
transition $\achain_n \rightarrow \achain_{n+1}$ and hence
$G_{n}^{(0)}=Q_{n}$.   
\end{proof} 

Using these results it was shown in \cite{csgone} that in CSG, the
transition probability from $c_n$ to $c_{n+1}$ with $\varpi$ precursor
elements and $m$ maximal elements takes on the simple form 
\begin{equation}
A_n=Q_n\sum_{k=0}^{m} (-1)^k \begin{pmatrix}
m\\k 
\end{pmatrix}\frac{1}{Q_{\varpi-m}}
\end{equation}

Before we end this section, we note that CSG dynamics
satisfies general covariance, the full dynamics is apriori {\it not} 
covariant, for the following reason. Since $\fA$ is generated by
finite set operations on naturally labelled causal sets in the set of
past finite naturally labelled causal sets $\Omega$. Thus the
subsets $\alpha \subset \Omega$ are not themselves covariant.  One way
of obtaining  a covariant event algebra \cite{observables} is to
extend the event algebra $\fA$ to its minimal sigma algebra
$\fS$. Since the measure $\mu$ on $\fA$ is classical, by the
Kolmogorov-Caratheodary-Hahn(KCH)  theorem it
extends to $\fS$. This means that the measure extends to covariant events 
$\fS$, which is of physical interest.

\section{Quantum Sequential growth}
\label{qsg.sec} 

We now turn to the quantisation of the classical sequential growth
models \cite{djs,sz}. As discussed in the introduction, by
quantisation we mean that the classical probability measure  $\mu$ in
the  
stochastic  triple  $(\Omega, \fA, \mu)$ must be replaced by its
quantum counterpart, i.e., the {\sl decoherence functional} or
equivalently, the {\sl  quantum measure}
\cite{qmeasureone}\cite{qmeasuretwo} which we now define.

Let $\Omega$ be a sample space of histories and
$\fA$ an associated  event algebra.
\begin{definition} A {\sl decoherence functional} $D :\fA\times \fA \rightarrow \mathbb{C}$ satisfies 
\begin{enumerate}
\item Hermiticity: for $\alpha ,\beta \in \fA$ $D(\alpha ,\beta )=D(\beta ,\alpha )^*$
\item Finite Biadditivity: for $\alpha_i ,\beta_j \in \fA$,  
  $D(\bigsqcup_i\alpha_i , \bigsqcup_j \beta_j)=\sum_{i,j} D(\alpha_i
  ,\beta_j)$, where $i$ and $j$ run over finite number of values
  and $\bigsqcup$ denotes disjoint union.    
\item Normalisation: $D(\Omega ,\Omega)=1$ 
\item Strong Positivity: For any finite collection of events $\a_i \in
  \fA (i = 1,...,N)$ the $N \times N$ matrix
 $D_{ij}\equiv D(\a_i,\a_j)$ is positive semi-definite, i.e.,
 $D_{ij}v^iv^j \geq 0$ for an arbitrary choice of  complex vector $v^i$.
\end{enumerate}
\end{definition}

\begin{definition} 
The {\sl  quantum  measure}
$\mu(\alpha): \fA \rightarrow \re^{\geq 0}$ satisfies $\mu(.) \geq
0$.  While it  need not satisfy additivity, i.e., in general 
\begin{equation}
\mu (\alpha\sqcup \beta) \neq \mu (\alpha)+ \mu
(\beta),  \label{nonadditive.eq} 
\end{equation}
it satisfies the {\sl quantum sum rule} 
\begin{equation}
\mu (\alpha\sqcup \beta \sqcup \gamma) =\mu (\alpha \sqcup \beta)+\mu
(\alpha \sqcup \gamma )+\mu (\beta \sqcup \gamma)-\mu (\alpha)-\mu
(\beta)-\mu (\gamma). 
\end{equation}
\end{definition}
The quantum measure is obtained from the decoherence functional by
identifying 
\begin{equation}
  \mu(\alpha) \equiv  D(\alpha,\alpha) \geq 0. 
 \end{equation}

One of the challenges of working with the quantum measure directly is
its non-additivity \eqref{nonadditive.eq}. In particular, there is no
known analogue of the KCH theorem for the decoherence
functional. However, as shown in \cite{histories,djs}, $D(.,.)$ can be
converted into an additive  {\sl  vector valued measure}
$\ket{.}$ on $\fA$ which is valued in a Histories Hilbert space
$\cH$. $\cH$ is itself 
constructed from $\fA$ using a GNS-type  construction, as follows.

Consider the set $\mathcal F$ of complex functions $f :\mA \rt
\mathbb{C}$ which have support on finitely many elements of
$\mA$ and let $\chi_\alpha$ correspond to the characteristic function
for every $ \alpha \in \mA$. $\mathcal F$ can then be vested with a
vector space structure  
\begin{eqnarray}
(f_1+f_2)(\alpha )&\equiv & f_1(\alpha ) +f_2 (\alpha )\\
(\lambda f_2 )(\alpha )&\equiv & \lambda f_2 (\alpha ), 
\end{eqnarray} 
with an inner product given by the decoherence functional 
\begin{equation}
\langle f_1 ,f_2\rangle \equiv \sum_{\alpha\in
  \mathcal{U}}\sum_{\beta\in \mathcal{U}}f_1(\alpha )^*D(\alpha ,\beta
)f_2(\beta ). 
\end{equation}
In particular, writing $\chi_\alpha=\ket{\alpha}$, we see
that
\begin{equation}
  \braket{\alpha}{\alpha}=D(\alpha,\alpha) 
\end{equation}
which is the quantum measure.  Taking the equivalence class of vectors
in $\mathcal F$ related by those of zero norm, and taking the Cauchy
completion, gives us the  Histories  Hilbert space $\cH$ \cite{djs}. 

\paragraph{Definition} \textit{ The quantum vector measure is a map $\ket\cdot :\mA\rt \cH$ 
such that 
\begin{enumerate}
\item For $\alpha\in \mA$, $\braket{\a}{\a} =\mu (\alpha )$, where $\mu$ is the quantum measure on $(\Omega , \mA)$
\item For any finite collection $\a_i\in\mA$, $i=1, \ldots n < \infty$   such that $\a_i \cap\a_j=\emptyset$ $\forall
  i\neq j$    
\begin{equation}
\ket{\bigsqcup_i^n \a_i}  =\sum_i^n \ket{\a_i}. 
\end{equation}
\end{enumerate}
}

While a classical probability measure on $\mA$ is (almost) guaranteed
to extend to its sigma algebra $\mS$, a  vector measure on $\fA$
extends to one on $\fS$ only if certain convergence conditions, given by the Caratheodary-Hahn-Kluvnek theorem
are satisfied \cite{diestel}. In complex CSG models, for example, the
complex measure does not always extend \cite{djs,sz}.

\subsection{Quantum sequential growth dynamics}
\label{qsgd.ssec}

As in the CSG models, in QSG the vector measure $\ket{.}: \fA \rightarrow \cH$  
 is finitely additive.  Starting with the state $\ket{\cyl{c_1}} = \ket{\Omega}$, every transition $\cT_n: c_n
\rightarrow c_{n+1}^i$ in $\cP$ is 
implemented by acting on $\ket{\cyl{c_n}}$ with a {\sl {transition
  operator}} $\hhT_n$, which generate the transition operator algebra
$\cA$.  The requirements of MSR, GC and QBC then lead to various constraints on
$\cA$.  

In \cite{djs,sz}  the simplest class of of QSG models were
constructed with $\cH \simeq \bC$, and hence $\cA$ is commutative. In this case the     
  decoherence functional simplifies to a product form 
  $D(\alpha,\beta)=A^*(\alpha)A(\beta)$. Implementing the three
  conditions on this dynamics results in a form for the transition
  operator that is very  similar to CSG
  models and thus can be studied more extensively.  It is our aim to
  extend this construction to $\cH$ of arbitrary
  dimensions, or equivalently,  to a transition operator algebra $\cA$
  that is  not
  necessarily commutative.

We will adopt the familial language of CSG wholesale for the transition
operators, and use the following notation: $\hG_n$ denotes a
gregarious operator, $\hTi_n$ the timid transition operator and
$\hhT_n$ a generic transition operator.

As in the case of the CSG, we wish to impose appropriate version of
the  MSR, GC and BC, thus constraining  the transition operators
$\hhT_n$. Any $\ket{c_n} $ can be expressed as a string of transition
operators acting on $\ket{\Omega}$
\begin{equation}
\ket{c_n} = \hhT_{n-1} \ldots \hhT_1 \ket{\Omega} \equiv
\hhT_{\gamma(c_n)} \ket{\Omega}.  
  \end{equation} 
where $\gamma(c_n)$ denotes the path in $\cP$ from the single element
causal set to $c_n$. 

\noindent {\bf Markov Sum Rule (MSR):}\\

As we have already seen, additivity of the measure ensures the
MSR. Explicitly,  
$\mA$ 
\begin{equation}
\cyl{c_n} = \bigsqcup_{i \in \cI(c_n)} {c_{n+1}^i}  \Rightarrow
\ket{c_n} =\sum_{j\in \mathcal{J}(c_n)}\ket{c_{n+1}^j}, 
\end{equation}
where we use the convention $\ket{c_n} \equiv \ket{\cyl{c_n}}$.  
This implies that 
\begin{equation} 
\ket{c_n} =\sum_{j \in \mathcal{J}(c_n)}\hhT_{n}^j \ket{c_n}
 \Rightarrow \sum_{j \in  \mathcal{J}(c_n)}\hhT_{n}^j = \mathds{I}  
\end{equation} 
 where $\mathbb{I}$ denotes the identity operator. 

\noindent {\bf General Covariance:}\\

The simplest implementation of GC on labelled poscau $\cP$, is to
require that $\ket{c_n}=\ket{c_n'}$ whenever $c_n' \sim c_n$ and
$c_{n+1}' \sim c_{n+1}$, or equivalently
$\hhT_{\gamma(c_n)} \ket{\Omega} =\hhT_{\gamma(c_n')} \ket{\Omega}
$. This implies that the inner product or quantum measure
$\braket{c_n}{c_n}=\braket{c_n'}{c_n'}$.   There is no operator
ordering ambiguity,  since $\hhT_{\gamma}$ is strictly time ordered.

\noindent {\bf Note:}  It is also possible to introduce a phase $\ket{c_n}=e^{i
  \Phi(c_n,c_n')}\ket{c_n'}$ since this will leave the quantum
measure invariant.  Because of the ensuing ambiguities, or generation
of new parameters   this will
result in, we do not consider this option henceforth.

In the next section we come to the main focus of our paper, namely on
the  implementation of Bell Causality in QSG models.

\subsection{Implementing Quantum Bell Causality} 
\label{qbc.ssec}

As discussed in the section \ref{sgd.ssec},  causality implies that
the transition operators associated with the related Bell 
pairs of transitions $\{\cT_n,\cT_n'\}, \{\cT_m,\cT_m' \}$ are 
related, i.e., that $\exists \, \, \mathcal{F}$ such that 
\begin{equation}
\cF(\hhT_n,\hhT_n',\hhT_m,\hhT_m')=0 \label{gBCGen.eq} 
\end{equation}
where the $\hhT_n,\hhT_n',
\hhT_m,\hhT_m'$ are the associated transition
operators. 

In CSG models, the BC condition results in a simple rule
\eqref{csgbc.eq} for the transition probabilities. However, 
because the related Bell pairs involve two different stages,
there are operator ordering ambiguities for the transition
operators. Any implementation of its quantum version, quantum Bell
causality (QBC)  requires us to make a specific
choice of operator ordering.  

Before making explicit choices of operator ordering, we prove 
with some assumptions on $\cF$, a fairly strong result in QSG,
analogous to   Lemma \ref{csgonelemma.lemma}  of  \cite{csgone}, which states that every transition
operator from $c_n$ can be obtained from the gregarious transition
operator from $c_n$.

Assume that 
\begin{enumerate}  
\item For all non timid transition operators $\hhT_n$,  with  $n>m$,
  $\cF$ can be re-expressed as 
\begin{equation}
\hhT_n=\tcF(\hhT_n',\hhT_m,\hhT_m') \label{gBC.eq} 
\end{equation}
\item For the gregarious transition operators $\hG_{n,m}$ and $n>m$
  \begin{equation}
\tcF(\hG_n,\sum_{j \in \bcJ(c_m)}\hhT_m^j,\hG_m)=\sum_{j\in \bcJ(c_m)}\tcF(\hG_n,\hhT_m^j,\hG_m) \label{gBCsum.eq} 
\end{equation}
where the index set $\bcJ(c_m)$ is over  non-gregarious transitions
from $c_m$. 
\end{enumerate} 
The requirement Eqn. \eqref{gBC.eq} is already very strong, since 
it means that every non-timid transition operator $\hhT_n$ is 
determined by the  $n^{th}$ stage  gregarious transition operator
$\hG_n$, i.e., 
\begin{equation}
  \hhT_n=\tcF(\hG_n,\hhT_m,\hG_m). 
\end{equation}
The MSR therefore implies that this is also true for the timid
transition operator $\hTi_n$. Moreover, inductively one sees that the earlier stage
transition operator  $\hhT_m$ in turn
depends on $\hG_m$ and so on. This implies the following: 
\begin{lemma} 
Let $\gamma(c_n)= \{c_1, 
c_2(c_n) \ldots c_n \}$ denote the unique labelled path in
$\cP$ to $c_n$. Then using GC,  any  transition operator $\hhT_n$ can
be constructed entirely from the gregarious operators
$\hG_m:\ket{c_m(c_n)}  \rightarrow \ket{c_{m+1}^g(c_n)}$, where $m\leq
n$.       
\end{lemma}

We now show that the  timid transitions in particular take on a
useful form. 
\begin{lemma}\label{GTA.lemma}
Let $c_n$ be an $n$ element causal set with $\mm(c_n)=\{e_{m_1}, e_{m_2}
,\ldots e_{m_\bmm}\}$ the set
of its maximal elements.  The timid transition
operator $\hTi_n$ takes the form 
\begin{equation}
\hTi_n=\id-\sum_{\ell =1}^\bmm \sum_{k_1<\cdots k_\ell}^\bmm
(-1)^{\ell
  -1} \biggl( \tilde{\mathcal{F}}(\hG_n,\mathbb{I}-\hG_{n-\ell}^{(k_1,
  \ldots, 
    k_\ell )}  ,\hG_{n-\ell}^{(k_1, \ldots, k_\ell )})\biggr)  - \hG_n
, \label{timidop.eq} 
\end{equation}
where $\hG_{n-\ell}^{(k_1, \ldots k_\ell
  )}$ denotes the gregarious transition operator from the $n-\ell$ element
causal set  $c_n\setminus \lbrace m_{k_1},\cdots ,m_{k_\ell}\rbrace$. 
\end{lemma}
\begin{proof}
  
Let $\hTi_n$ denote the timid transition operator from $c_n$ and let 
$\hhT_n^j$ denote a non-timid operator, with $j\in \bcJ(c_n)$, 
the index set $\cJ(c_n)$ without the timid transition.  Using the MSR 
\begin{equation}
  \hTi_n=\mathbb{I}- \sum_{i \in \bcJ(c_n)} \hhT_n^i  \Rightarrow
  \hTi_n\ket{c_n} =\ket{c_n} - \sum_{i \in \bcJ(c_n)} 
  \ket{c_{n+1}^i}  \label{timidopone.eq} 
\end{equation}

From Eqn. \eqref{zetal.eq} the last term can be rewritten as 
\begin{equation} \sum_{j \in \bcJ(c_n)} 
  \ket{c_{n+1}^j} = \ket{\bigsqcup_{j \in \bcJ(c_n)} c_{n+1}^j} =
  \ket{\bigcup_{k=1}^\bmm \zeta_k}. 
\end{equation}

From Eqn \eqref{gmuzeta.eq} 
\begin{eqnarray} 
\ket{\bigcup_{k=1}^{\bmm} \zeta_k}  &=&  \sum_{k=1}^{\bmm}
\ket{\zeta_k}  -\sum_{k_1<k_2}^{\bmm}\ket{\zeta_{k_1}  \bigcap
\zeta_{k_2} } + \cdots + (-1)^{\bmm-1}
                                        \ket{\bigcap_{k=1}^{\bmm} \zeta_k} \nonumber \\
 &=& \sum_{\ell=1}^\bmm  \sum_{k_1<k_2, \ldots k_\ell }^\bmm (-1)^{\ell-1}
                                        \ket{\zeta_{k_1} \bigcap
     \zeta_{k_2} \ldots \bigcap \zeta_{k_\ell} }
\label{unionzeta.eq} 
\end{eqnarray} 
Now the set $\bigl(\zeta_{k_1} \cap \ldots \cap \zeta_{k_\ell}\bigr)$
is also the (disjoint) union  
of all the cylinder sets of the children of $c_n$ for which the subset of maximal elements,   $\{
e_{m_{k_1}}, \ldots, e_{m_{k_\ell}}\}$ is in the spectator set. Using Eqn. \eqref{gmuintzeta.eq} 
\begin{equation}
\ket{\zeta_{k_1} \cap \ldots \cap \zeta_{k_\ell}} =\sum_{j \in
  \cJ^{(k_1,\ldots k_\ell)}}  
\ket{c_{n+1}^{j}} = \biggl( \sum_{j \in
  \cJ^{(k_1,\ldots k_\ell)}} 
\hhT_n^{j}\biggr) \ket{c_n}, \label{intzeta.eq}
\end{equation}
where $\cJ^{(k_1,\ldots k_\ell)}$ labels the set of children of $c_n$
  with   $\{e_{m_{k_1}}, \ldots, e_{m_{k_\ell}}\}$  in the spectator set. 

Consider a  transition  $\cT_n^j:c_n \rightarrow c_{n+1}^j$ with $\{e_{m_{k_1}}, \ldots e_{m_{k_\ell}}\}$ in its 
spectator set 
$j \in \cJ^{(k_1,\ldots, k_\ell)}$, and its Bell partner
$\cT_{n-\ell}^{j(k_1, \ldots k_\ell)}$ which takes
$c_{n-\ell}^{(k_1,\ldots, k_\ell)} \equiv c_n \backslash
\{e_{m_{k_1}}, \ldots, 
e_{m_{k_\ell}}\} $ to $c_{n-\ell +1}^{j,(k_1, \ldots, k_\ell)}$, $j \in
\cJ(k_1, \ldots, k_\ell)$.  The index set $\cJ(k_1, \ldots, k_\ell)$
therefore labels {\it all}
transitions from   $c_{n-\ell}^{(k_1,\ldots, k_\ell)}$. 

For a non-gregarious transition $\cT_n^j$, consider the Bell pair
$(\cT_n^j, \cG_n)$ which lies in the same  Bell
family as $(\cT_{n-\ell}^{j(k_1,\ldots, k_\ell )}, \cG_{n-\ell}
^{(k_1\cdots k_\ell )})$. If $\hhT_n^j$ denotes the transition
operator associated with $\mathcal{T}_n^j$, and $\hhT_{n-\ell}^{j(k_1\cdots
  k_\ell )}$ with $\cT_{n-\ell}^{j(k_1\cdots
  k_\ell )}$, respectively,  the quantum Bell causality (QBC) relation
Eqn. \eqref{gBC.eq} implies that 
\begin{equation}
\hhT_n^j = \tcF(\hG_n, \hhT_{n-\ell} ^{j(k_1\cdots
  k_\ell )}, \hG_{n-\ell}^{(k_1\cdots
  k_\ell )}). \label{nongreggbc.eq} 
 \end{equation} 
Inserting into Eqn. \eqref{intzeta.eq} after separating out the gregarious transition
operator and using the property
Eqn. \eqref{gBCsum.eq}  
\begin{eqnarray} 
 \ket{\zeta_{k_1} \cap \ldots \cap \zeta_{k_\ell}} & =& \biggl(
 \sum_{j\in \bcJ^{(k_1,\ldots k_\ell)}} \hhT_n^j+ \hG_n
                                                     \biggr)\ket{c_n}
                                                     \nonumber \\
  &=& \biggl(
 \sum_{j\in \bcJ^{(k_1,\ldots k_\ell)}} \tcF(\hG_n, \hhT_{n-\ell} ^{j(k_1\cdots
  k_\ell )}, \hG_{n-\ell}^{(k_1\cdots
  k_\ell )}) + \hG_n\biggr)\ket{c_n}
      \nonumber \\
  &=& \biggl( \tcF(\hG_n, \sum_{j\in \bcJ^{(k_1,\ldots k_\ell)}} \hhT_{n-\ell} ^{j(k_1\cdots
  k_\ell )}, \hG_{n-\ell}^{(k_1\cdots
  k_\ell )}) + \hG_n\biggr)\ket{c_n}, 
\end{eqnarray}
where now $\bcJ^{(k_1,\ldots k_\ell)}$ labels all but the gregarious
transition from $c_{n-\ell}^{(k_1,\ldots k_\ell)}$. Thus
\begin{equation}
  \sum_{j\in \bcJ^{(k_1,\ldots k_\ell)}} \hhT_{n-\ell} ^{j(k_1\cdots
  k_\ell )} = \id - G_{n-\ell}^{(k_1,\ldots k_\ell)}, 
\end{equation}
i.e.,
\begin{equation}
\ket{\zeta_{k_1} \cap \ldots \cap \zeta_{k_\ell}} =\biggl( \tcF(\hG_n, \id - G_{n-\ell}^{(k_1,\ldots k_\ell)} , \hG_{n-\ell}^{(k_1\cdots
  k_\ell )}) + \hG_n)\biggr) \ket{c_n}. 
  \end{equation} 
  Thus, Eqn \eqref{unionzeta.eq} can be expressed as
  \begin{equation}
\ket{\bigcup_{k=1}^{\bmm} \zeta_k}  =\sum_{\ell=1}^\bmm
\sum_{k_1<k_2, \ldots i_\ell }^\bmm (-1)^{\ell-1} \biggl( \tcF(\hG_n, \id - G_{n-\ell}^{(k_1,\ldots k_\ell)} , \hG_{n-\ell}^{(k_1\cdots
  k_\ell )}) + \hG_n)\biggr) \ket{c_n}. 
    \end{equation}

Inserting into Eqn. \eqref{timidopone.eq}
  \begin{eqnarray}
\hTi_n&=&\id-\sum_{\ell =1}^\bmm \sum_{k_1<\ldots k_\ell }^\bmm
    (-1)^{\ell
    -1}\biggl(\tilde{\mathcal{F}}(\hG_n,\id-\hG_{n-\ell}^{(k_1\ldots
    k_\ell )}  ,\hG_{n-\ell}^{(k_1\cdots k_\ell )})\biggr)
    -\sum_{\ell=1}^\bmm \binom{\bmm}{\ell} 
    (-1)^{\ell -1}  \hG_n \nonumber\\
&=&\mathbb{I}-\sum_{\ell=1}^\bmm\sum_{k_1<\ldots k_\ell }^\bmm
    (-1)^{\ell
    -1}\biggl(\tilde{\mathcal{F}}(\hG_n,\id-\hG_{n-\ell}^{(k_1,\ldots
    k_\ell )}  ,\hG_{n-\ell}^{(k_1, \ldots k_\ell )})\biggr) -\hG_n \label{timid.eq}
\end{eqnarray}

\end{proof}

Along with Eqn. \eqref{nongreggbc.eq}  this gives  a quantum
generalisation of  Lemma \ref{csgonelemmaone.lemma} of \cite{csgone},
namely that all transition operators  from $c_n$ can be obtained in
from the gregarious operator $\h G_n$ from $c_n$ (and lower
stage operators). However, to make 
further progress we must make  explicit choices for the QBC condition Eqn. \eqref{gBCGen.eq}.

Given the operators $\hhT_n,\hhT_n',\hhT_m,\hhT_m'$ that partake in
the Bell causality condition, we remind ourselves that  the
classical (and hence 
commutative) BC condition can expressed either as  $\frac{P_n}{P_n'}= \frac{P_m}{P_m'}$
for non-vanishing transition probabilities $P_n,P_n', P_m, P_m'$, or more generally as the
product Bell causality (PBC) \cite{dsobs} $  P_n P_m' = P_n' P_m $
which can be implemented even when the transition probabilities vanish
\cite{rv,dsobs}. Implementing the BC  as a quantum condition means making operator 
ordering choices.

We now consider  three linear
implementations of QBC below.  Note that we  limit ourselves to transition
operators that can be inverted, i.e., $\hhT_n^{-1}$ always exists.  
   
\begin{description}
\item[1) Time Ordered (TOBC):] For $n >m$ let 
\begin{equation}
\hhT_n\hhT_m'=\hhT_n'\hhT_m
\end{equation}
\item[2) Non-Time Ordered (NTOBC):] 
\begin{equation}
\hhT_n \hhT_n'^{-1}=\hhT_m\hhT_m'^{-1} \Leftrightarrow \hhT_n' \hhT_n^{-1}=\hhT_m'\hhT_m^{-1} 
\end{equation}
\item[3)Causal Past Ordered (CPOBC):]  
\begin{equation}
\hhT_n\hhT_m'=\hhT_m\hhT_n', 
\end{equation}
when the precursor set $\pre(\cT_n)=\pre(\cT_m)$  is strictly larger than
$\pre(\cT_n')=\pre(\cT_m')$. When $\pre(\cT_n)=\pre(\cT_n')$  there is
no basis for the ordering 
\begin{eqnarray}
  \hhT_n\hhT_m' = \hhT_m\hhT_n' &\quad \& & \hhT_n'\hhT_m= \hhT_m'\hhT_n
                                         \nonumber \\ 
 \Rightarrow     \hhT_n = \hhT_m\hhT_n'  \hhT_m'^{-1}=  \hhT_m'^{-1}
  \hhT_n'\hhT_m                & \Rightarrow &   \hhT_m'\hhT_m \hhT_n'  =  
  \hhT_n'\hhT_m \hhT_m'. 
\end{eqnarray}
The  last equality is an additional constraint on the 
  algebra (note that it cannot be case as a commutation 
  relation).
\end{description}

Note that these implementations are all linear, and reduce to the classical BC
condition. Hence even if there are other ways to implement these
conditions with inclusions of non-linear terms at higher order, to
leading order we should expect these terms to exist. It is easy to
check that all three conditions satisfy Eqns \eqref{gBCGen.eq},
\eqref{gBC.eq} and \eqref{gBCsum.eq}.

The question that we seek an answer to is whether the  algebra $\cA$
generated by the transition operators is non-commutative. 
Note that  requiring the transition operators to be invertible  does not contradict the
Gelfand Mazur theorem \cite{rudin}  which states that if
{\it every}  element of the algebra is  invertible,  then it is 
 isomorphic to $\bC$.  Even though the transition operators may not
 themselves be invertible, the algebra $\cA$ could contain operators
 that are non-invertible elements which are not themselves transition amplitudes.

In the following subsections we consider each of these  three possible
quantum  Bell causality conditions. Our strategy will be to use the 
atomisation   construction of  Lemma \ref{csgonelemma.lemma} of
\cite{csgone}, now using  GC and the QBC to relate the
  $\hG_n$  to $\hQ_n$. This is also supplemented by We will then use  Lemma 
\ref{GTA.lemma}  for the general form of the transition
operators to express them in terms of the $\hQ_n$ to the extent
possible.  Because of the assumption of non-commutativity, the
atomisation   pathways are not unique and lead to new relations
between the operators.

Before proceeding we note that if the  condition of invertibility is relaxed,
it is relatively easy to see that the dynamics generates an ever increasing number of new parameters at every
stage. We deem this as unsatisfactory both from a practical, solution
seeking point of view, but also from a physics point of view since the
proliferation of parameters signals a type of unpredictability.

\subsection{Time Ordered Bell Causality}\label{sec1}
\label{tobc.ssec} 

The first thing is to find the  relation between $\hG_n$ and $\hQ_n$
in the manner of  Lemma \ref{csgonelemma.lemma}   of \cite{csgone}
(see Fig \ref{atomisation.fig}).  Let
$\hB_k^{(i)}$ be the transition operator associated with the
  transition $\cB_k^{(i)}$, and $\hG_k^{(i)}$ the transition operator 
  corresponding to the gregarious transition $\cG_k^{(i)}$.

Using the related Bell pairs of transition
operators 
$\{\hB_n^{(1)}, \hG_n^{(1)} \}$  and $\{\hB_{n-1}^{(0)},
\hG_{n-1}^{(0)} \}$ in Fig \ref{atomisation.fig}, TOBC gives 
\begin{equation}
\hB_{n}^{(1)}\hG_{n-1}^{(0)}=\hG_{n}^{(1)}\hB_{n-1}^{(0)}, 
\end{equation}
while from  GC 
\begin{equation}
\hG_{n}^{(0)}\hB_{n-1}^{(0)}=\hB_{n}^{(1)}\hG_{n-1}^{(0)}. 
\end{equation}
Eliminating $\hB_{n-1}^{(0)}$ we find 
\begin{equation}
  \hG_{n}^{(0)}= \hG_{n}^{(1)}. 
  \end{equation} 
Using atomisation, this means that $\hG_n^{(0)}=\hQ_n$ for all
$n$ and all $c_n$.   Thus we have proved the TOBC version of Lemma \ref{csgonelemma.lemma}
of \cite{csgone}: 
\begin{lemma} 
Every gregarious transition operator $\hG_n =
\hQ_n$. \label{lemmatwotobc.lemma} 
\end{lemma}

We now focus on the transitions from the antichain $\achain_n$. In
particular, consider the transition operator $\hA_n^{(1)} :\ket{ a_n}  \rightarrow \ket{\bchain_{n+1}}$
where the new element $e_{n+1} $ in $\bchain_{n+1}$ has a single
element to its past.  For any $1 \leq m < n$,  $\hA_n^{(1)}$ has a  Bell
partner $\hA_m^{(1)}: \ket{\achain_m} \rightarrow \ket{\bchain_{m+1}}$. In
particular, for $m=1$,   
$\hA_1^{(1)}: \ket{a_1} \rightarrow \ket{\chain_2}$, where $\chain_2$ is the
$2$-element chain and $\hA_1^{(1)}=\id - \hQ_1$.  Using the related Bell pairs $(\hA_n^{(1)},\hQ_n)$ and
$(\hA_m^{(1)},\hQ_m)$ for $m < n $ and the TOBC, 
\begin{equation}
\hA_n^{(1)}\hQ_m=\hQ_n\hA_m^{(1)} \Rightarrow \hA_n^{(1)}= \hQ_n\hA_m^{(1)}
\hQ_m^{-1}.  \label{bnqm.eq} 
\end{equation} 
Now for $m=1$, for all $n>1$  
\begin{equation} 
\hA_n^{(1)}=\hQ_n(\hQ_1^{-1} - \id). 
\end{equation}
Reinserting into  Eqn. \eqref{bnqm.eq} both for $n$ and for $m$ this
means that 
\begin{equation}
 \hQ_n(\hQ_1^{-1} - \id) = \hQ_n \hQ_m(\hQ_1^{-1} - \id) \hQ_m^{-1}
 \Rightarrow [\hQ_1, \hQ_m]=0.
\end{equation}
In order to prove that this is the case more generally, we note that
the TOBC satisfies both conditions Eqn. \eqref{gBC.eq} and
Eqn. \eqref{gBCsum.eq}, i.e., for every Bell related pair  $\{
\hA_n,\hQ_n\}, \{\hA_m\hQ_m \}$, 
\begin{equation}
\hA_n=\hQ_n\hA_m\hQ_m^{-1} \equiv
\tcF(\hQ_n,\hA_m,\hQ_m) \label{tobccF.eq} 
\end{equation}
and therefore
\begin{equation} 
\tcF(\hQ_n, \sum_{j\in \bcJ(c_m)} \hA_m^j,\hQ_m) = \hQ_n\biggl(
\sum_{j\in \bcJ(c_m)} \hA_m^j \biggr) \hQ_m^{-1} = 
\sum_{j\in \bcJ(c_m)}  \hQ_n\hA_m^j  \hQ_m^{-1}. 
\end{equation} 

Using Lemmas \ref{GTA.lemma}  and \ref{lemmatwotobc.lemma}  and noting
that  all $G_n^{i_1\cdots i_\ell)}=Q_n$,   the timid
operator at stage $n$ can be expressed as 
\begin{eqnarray}
\hTi_n&=&\id -\sum_{\ell =1}^\bmm\sum_{i_1<\cdots i_\ell}^\bmm (-1)^{\ell
  -1}\hQ_n \bigl( \id-\hQ_{n-\ell} \bigr) \hQ_{n-\ell}^{-1} \, -\,
  \hQ_n \nonumber \\
  &=& \id -\sum_{\ell =1}^\bmm\sum_{i_1<\cdots i_\ell}^\bmm (-1)^{\ell
  -1}\hQ_n \bigl( \id-\hQ_{n-\ell} \bigr) \hQ_{n-\ell}^{-1} \, -\,
      \hQ_n \nonumber \\
  &=& \hQ_n\sum_{\ell=0}^\bmm (-1)^\ell \binom{\bmm}{\ell}
  \hQ_{n-\ell}^{-1} \label{timidtobc.eq} 
\end{eqnarray}

We now find that the $\hQ_n$'s all commute with each other.  To show this it suffices
to look at the set of transitions from  $\achain_n$, which are 
characterised completely by the number of precursor elements. Let
$\hA_n^{(k)}$ be the transition operator from $\achain_n$ to its child
$\achain^{(k)}_{n+1}$ which has $k$ elements in the precursor set of
the new element $e_{n+1}$.  

Wlog, let $k< n < m$.

We have the Bell family of pairs: $\{\hA_m^{(k)}, \hQ_m \}, \{\hA_n^{(k)},
\hQ_n \}, \{\hA_k^{(k)}, \hQ_k \}$.  Note that $\hA_k^{(k)}$ is the
timid transition, and hence $\hQ_\ell$ for $\ell < k$ doesn't belong
to this family. Choosing the first two pairs TOBC gives 
\begin{equation}
\hA_m^{(k)} \hQ_n =  \hQ_m \hA_n^{(k)} \Rightarrow  \hA_m^{(k)} =
\hQ_m \hA_n^{(k)} \hQ_n^{-1}.  \label{Lambdamn.eq} 
  \end{equation} 
  Using
  \begin{equation}
\hA_{m,n}^{(k)} =
\hQ_{m,n} \hA_k^{(k)} \hQ_k^{-1}, 
    \end{equation} 
    Eqn. \eqref{Lambdamn.eq} simplifies to
\begin{eqnarray} 
  \hA_k^{(k)} \hQ_k^{-1} \hQ_n &= &  \hQ_n \hA_k^{(k)} \hQ_k^{-1}
                                    \nonumber \\
  \Rightarrow \sum_{j=0}^k (-1)^j \binom{k}{j} \hQ_k \hQ_{k-j}^{-1} \hQ_k^{-1}
  \hQ_n &=&  \sum_{j=0}^k (-1)^j \binom{k}{j} \hQ_n \hQ_k
            \hQ_{k-j}^{-1} \hQ_k^{-1}, \label{qnqk.eq} 
  \end{eqnarray} 
where we have used the expression Eqn. \ref{timidtobc.eq} for $\hA_k^{(k)}$ and
the fact that $\bmm=k$.    

For $k=1$, the above expression reduces for all $n$ to
\begin{equation}
\hQ_1^{-1} \hQ_n - \hQ_n = \hQ_n \hQ_1^{-1} -\hQ_n  \Rightarrow
[\hQ_1, \hQ_n]=0.  \label{qoneqn.eq} 
\end{equation}
Thus, the subalgebra $\cQ$ has a non-trivial center. In fact, this in
turn 
implies that  for $k=2$ 
\begin{equation}
\hQ_2^{-1} \hQ_n - 2\hQ_1^{-1} \hQ_n + \hQ_n = \hQ_n \hQ_2^{-1} -\hQ_n
- 2 \hQ_n \hQ_1^{-1} + \hQ_n\Rightarrow
[\hQ_2, \hQ_n]=0. 
\end{equation}
This suggests an inductive generalisation.

Assume that $[\hQ_{k'}, \hQ_n]=0$ for all $n$ and  $k'<k$.   Expanding
Eqn. \eqref{qnqk.eq} 
\begin{eqnarray}
 \hQ_k^{-1} \hQ_n - k  \hQ_k \hQ_{k-1} \hQ_{k}^{-1} \hQ_n +
  \frac{k(k-1)}{2} \hQ_k \hQ_{k-2} \hQ_{k}^{-1} \hQ_n \ldots
  (-1)^{k-1} k  \hQ_k \hQ_{1} \hQ_{k}^{-1} \hQ_n + (-1)^k \hQ_n
  \nonumber \\ = \hQ_n \hQ_k^{-1} \hQ_n - k  \hQ_k \hQ_{k-1} \hQ_{k}^{-1} +
  \frac{k(k-1)}{2} \hQ_n\hQ_k \hQ_{k-2} \hQ_{k}^{-1} \ldots
  (-1)^{k-1} k  \hQ_n\hQ_k \hQ_{1} \hQ_{k}^{-1}+ (-1)^k \hQ_n.
\end{eqnarray}
Since for all $j<k$
\begin{equation}
  \hQ_k \hQ_{k-j} \hQ_{k}^{-1} \hQ_n = \hQ_{k-j}  \hQ_n =\hQ_n  \hQ_{k-j}  
\end{equation}  
all but the first terms on both sides of the equation are the same and
hence we are left with 
\begin{equation}
  \hQ_k^{-1} \hQ_n= \hQ_n \hQ_k^{-1} \Rightarrow [\hQ_k, \hQ_n]=0. 
\end{equation} 

Next, using Eqn. \eqref{tobccF.eq} and the expression
Eqn. \eqref{timidtobc.eq},  we see that for $\hA_n$ a non-timid transition
operator and  $\hA_m=\h T_m$ a timid transition operator,  
\begin{equation}
  \hA_n=\hQ_n\h T_m\hQ_m^{-1} = \hQ_n \sum_{\ell=0}^\bmm (-1)^\ell \binom{\bmm}{\ell}
  \hQ_{n-\ell}^{-1}. 
\end{equation} 
Since every   transition  can be expressed in terms of the $\hQ_n$'s
which themselves commute we have shown that 

\begin{lemma}
The algebra of finitely generated transition operators $\mathcal{A}$
using the TOBC is commutative. 
\end{lemma}

 \subsection{Non-Time Ordered Bell Causality (NTOBC)}  
 \label{ntobc.ssec}
 
For Bell related pairs $\{\hA_n, \hB_n \}, \{\hA_m, \hB_m \}$ the
NTOBC condition is 
\begin{equation}
\hA_n \hB_n^{-1}=\hA_m\hB_m^{-1}. 
\end{equation}

Thus, 
\begin{equation}
\hA_n=\tcF(\hB_n,\hA_m,\hB_m) =   \hA_m\hB_m^{-1} \hB_n. 
\end{equation} 
Inserting into Eqn. \eqref{timid.eq} (Lemma \ref{GTA.lemma}) the timid
transition simplifies to
\begin{eqnarray} 
\hTi_n & = &  \id - \sum_{\ell =1}^m \sum_{k_1 < k_2 \ldots k_\ell}^m
(-1)^{\ell -1} \bigl(\id - \hG_{n-\ell}^{(k_1,k_2, \ldots k_\ell)} \bigr)
\bigl(\hG_{n-\ell}^{(k_1,k_2, \ldots k_\ell)}\bigr)^{-1} \hG_n -\hG_n
             \nonumber \\
  &=&  \sum_{\ell =0}^m (-1)^{\ell -1} \bigl(\hG_{n-\ell}^{(k_1,k_2,
      \ldots k_\ell)}\bigr)^{-1} \hG_n. 
 \end{eqnarray}

As before, we begin by examining the gregarious transition operators
at stage $n$ using atomisation (Figure \ref{atomisation.fig}).  From
GC and NTOBC, we get 
\begin{equation}
\hG_{n}^{(0)} \hB_{n-1}^{(0)}=\hB_{n}^{(1)} \hG_{n-1}^{(0)},  \quad 
\hB_{n}^{(1)}(\hG_{n}^{(1)})^ {-1}=\hB_{n-1}^{(0)}
(\hG_{n-1}^{(0)})^{-1}, \label{GCNTOBC.eq}
\end{equation}
which together imply that
\begin{equation}
\hG_{n}^{(0)} = \hB_{n}^{(1)} \hG_{n}^{(1)} (\hB_{n}^{(1)})^{-1}
\end{equation}
i.e.,  these  gregarious transition operators are related by
conjugation.  Proceeding this way, we see
that if the operators $\hB_{n}^{(1)}, \ldots  \hB_{n}^{(k)}$  lead to
atomisation  (see  Fig \ref{atomisation.fig}), then $\hG_{n}^{(0)}$ is related to
$\hQ_{n}$ by the conjugation 
\begin{equation}
\hG_{n}^{(0)} = \hB_{n}^{(1)} \ldots
\hB_{n}^{(k)} \hQ_{n} (\hB_{n}^{(k)}) ^{-1} \ldots
(\hB_{n}^{(1)})^{-1} \equiv  \hS_{n}^{(\alpha)} \hQ_{n}
(\hS_{n}^{(\alpha)})^{-1}. \label{conjugate.eq}  
\end{equation}
The operator $\hS_{n+1}^{(\alpha)}$ depends  on the
atomisation path  $\alpha$, which 
need not be unique. If   $\hS_{n}^{(\alpha)}$ and $\hS_{n}^{(\beta)}$
correspond to taking two paths $\alpha$ and $\beta$  that are distinct, this implies that for all $n$
\begin{equation}
\bigl[(\hS_{n}^{(\alpha)})^{-1}\hS_{n}^{(\beta)},
\hQ_{n}\bigr]=0. 
\end{equation} 

We can further use the decimation process (see Fig. \ref{AD.fig}) to express the conjugation
in terms of lower stage operators, i.e.  with $m<n$.  To begin with we
can replace the $\h B_n^{(1)}$ operators in Eqn. \eqref{GCNTOBC.eq} to
get 
\begin{equation}
  \hG_n^{(0)}=\hB_{n-1}^{(0)}  (\hG_{n-1}^{(0)})^{-1}\hG_{n}^{(1)}
  (\hG_{n-1}^{(0)}) (\hB_{n-1}^{(0)})^{-1}. 
\end{equation} 
Proceeding this way, along the atomisation path, we therefore obtain 
\begin{eqnarray} 
  \hG_n^{(0)} &= &  \h \tS^{(\alpha)} \hQ_{n} \bigl(\h \tS^{(\alpha)}\bigr)^{-1},
  \nonumber \\
\h\tS^{(\alpha)} &\equiv& 
 \hB_{n-1}^{(0)}  (\hG_{n-1}^{(0)})^{-1}\hB_{n-1}^{(1)} (\hG_{n-1}^{(1)})^{-1} \ldots \hB_{n-1}^{(k)} (\hG_{n-1}^{(k)})^{-1}. \label{decimation.eq}
\end{eqnarray} 
In turn, using NTOBC, 
\begin{eqnarray} 
  \hB_{n-1}^{(i)}  (\hG_{n-1}^{(i)})^{-1} &=&   \hB_{n-i-1}^{(0)}
                                              (\hG_{n-i-1}^{(0)})^{-1} \nonumber \\
 \h \tS^{(\alpha)} &\equiv& 
 \hB_{n-1}^{(0)}  (\hG_{n-1}^{(0)})^{-1}\hB_{n-2}^{(0)} (\hG_{n-2}^{(0)})^{-1} \ldots \hB_{n-k-1}^{(0)} (\hG_{n-k-1}^{(0)})^{-1}. 
\end{eqnarray} 
Thus,
\begin{equation}
\ket{c_n} = \hB_{n-1}^{(0)} \hB_{n-2}^{(0)} \ldots \hB_{n-k-1}^{(0)}
\hQ_{n-k} \ldots \hQ_1 \ket{\Omega} =\h \Gamma^{\alpha}   \ket{\Omega}
  \end{equation} 
The path $\alpha$ therefore also determines the operator  $\h\Gamma^{\alpha}$.

As in the case of TOBC, we consider the transition from the
$n$-antichain $\achain_n$ to its child  $\achain_{n+1}^{(k,i)}$, where $k$
denotes the number of elements in the precursor set and $i$  runs over
$\binom{n}{k}$ values (for the labelled case). Since the number of maximal elements $\bmm=n$,
from Eqn. \eqref{timid.eq}  the timid transition operator
$\hA_n^{(n)}$ for $\achain_n$ simplifies to
\begin{equation}
\hA_n^{(n)}=\sum_{\ell=0}^n (-1)^\ell \binom{n}{\ell} (\hQ_{n-\ell})^{-1}
\hQ_n. \label{ntobctimid.eq} 
\end{equation} 
Let   $\hA_n^{(k)}:\achain_n \rightarrow
\achain_{n+1}^{(k,i)}$. We then have the related Bell pairs
$\{\hA_n^{(k)}, \hQ_n \}$ and  $\{\hA_{k}^{(k)}, \hQ_{k} \}$
which, using Eqn. \eqref{ntobctimid.eq} implies that
\begin{equation}
\hA_n^{(k)}=\sum_{r=0}^{k}(-1)^r \binom{k}{r} \hQ_{k-r}^{-1}
\hQ_n.  \label{ntobcnontimid.eq} 
\end{equation}

Consider the two  transition operators  $\hA_n^{(k)}:\achain_n \rightarrow
\achain_{n+1}^{(k,i)}$ and $\hA_n^{(m)}:\achain_n \rightarrow
\achain_{n+1}^{(m,j)}$, where $k \neq m$. For ease of notation, let us fix the indices $i=i_1,
j=j_1$ for now and suppress these labels.  Next, let $\hB_{n+1}^{(k,m)}:
\achain_{n+1}^{(k)} \rightarrow c_{n+2}^{(k,m)}$, where the precursor set 
$\pre(c_{n+2}^{(k,m)})=\pre(c_{n+1}^{(m)})$ and hence does not
contain the element $e_{n+1}^{(k)} \in \achain_{n+1}^{(k)}$. We similarly
define $\hB_{n+1}^{(m,k)}:
\achain_{n+1}^{(m)} \rightarrow c_{n+2}^{(m,k)}$. Clearly,
$c_{n+2}^{(m,k)} \sim c_{n+2}^{(k,m)}$,  which by  GC
gives 
\begin{equation}
\hB_{n+1}^{(k,m)} \hA_{n}^{(k)}=\hB_{n+1}^{(m,k)}
\hA_{n}^{(m)}. \label{ntobcone.eq} 
  \end{equation} 

In addition, consider the gregarious transition operators  $\hG_{n+1}^{(k)}:
\achain_{n+1}^{(k)} \rightarrow c_{n+2}^{(k,g)}$, $\hG_{n+1}^{(m)}:
\achain_{n+1}^{(m)} \rightarrow c_{n+2}^{(m,g)}$, where $g$ denotes the
gregarious transition. This gives the related Bell pairs of
operators $\{\hB_{n+1}^{(k,m)}, \hG_{n+1}^{(k)}\}$ and
$\{\hA_{n}^{(m)}, \hQ_{n}\}$ which from NTOBC gives 
\begin{equation}
\hB_{n+1}^{(k,m)} = \hA_{n}^{(m)} \hQ_{n}^{-1} \hG_{n+1}^{(k)} \label{ntobctwo.eq} 
  \end{equation} 
and similarly, the related Bell pairs of
operators $\{\hB_{n+1}^{(m,k)}, \hG_{n+1}^{(m)}\}$ and
$\{\hA_{n}^{(k)}, \hQ_{n}\}$ gives 
\begin{equation}
\hB_{n+1}^{(m,k)} = \hA_{n}^{(k)}\hQ_{n}^{-1} \hG_{n+1}^{(m)}.  \label{ntobcthree.eq} 
  \end{equation} 
Inserting Eqns \eqref{ntobctwo.eq} and \eqref{ntobcthree.eq} into
Eqn. \eqref{ntobcone.eq}, we find that
\begin{equation}
\hA_{n}^{(m)}  \hQ_{n}^{-1} \hG_{n+1}^{(k)}
 \hA_{n}^{(k)}=\hA_{n}^{(k)} \hQ_{n}^{-1}\hG_{n+1}^{(m)}\hA_{n}^{(m)}
  . \label{ntobcfour.eq} 
\end{equation} 

Next, from Eqn. \eqref{conjugate.eq} the gregarious transition
$\hG_{n+1}^{(k)}$ is conjugate to $\hQ_{n+1}$ via
the transition operator $\hA_{n+1}^{(k)}: \achain_{n+1} \rightarrow \achain_{n+1}^{(k)}$, i.e., 
\begin{equation}
\hG_{n+1}^{(k)}=\hA_{n+1}^{(k)}\hQ_{n+1}(\hA_{n+1}^{(k)})^{-1}.  \label{ntobcfive.eq} 
  \end{equation} 
Using the related Bell pairs $\{\hA_{n+1}^{(k)},\hQ_{n+1}\}$ and
$\{\hA_n^{(k)},\hQ_{n}\}$ moreover,
\begin{equation}
  \hA_{n+1}^{(k)} = \hA_n^{(k)}\hQ_{n}^{-1} \hQ_{n+1}
\end{equation}
Inserting into Eqn. \eqref{ntobcfive.eq} we find   
\begin{equation} 
  \hG_{n+1}^{(k)}= \hA_n^{(k)}\hQ_{n}^{-1} \hQ_{n+1} \hQ_{n}
  (\hA_n^{(k)})^{-1}. 
\end{equation}
Reinserting this expression for $k$ and $m$ into
Eqn. \eqref{ntobcfour.eq} we find the commutation relation 
\begin{equation}
\bigl[ \hA_n^{(k)} \hQ_n^{-1}, \hA_n^{(m)}
\hQ_n^{-1}\bigr]. \label{ntobccommrel.eq} 
\end{equation}

At stage $n=1$, as always,  we have the two transitions operators $\hQ_1$  and
$\id - \hQ_1$  as shown in Fig \ref{cone.fig}. 
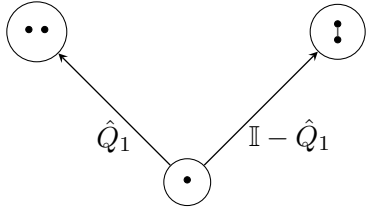
\begin{figure}[H]
 \centering
 \begin{tikzpicture}[-stealth]
  \def\ystep{2cm}
  \def\xstep{0.5cm}
  \begin{scope}[nodes={draw, thin, circle, minimum size=0.5cm}]
 \node (C0) at (0, 0) {\pcauset{1}};
 \node (C21) at (-4*\xstep, 1*\ystep) {\pcauset{2,1}};
 \node (C12) at ( 4*\xstep, 1*\ystep) {\pcauset{1,2}};
 \end{scope}
 \begin{scope}[prob arrow]
  \drawprobleft{C0}{\hQ_1}{C21} 
 \drawprobright{C0}{\qquad \id -\hQ_1}{C12}
 \end{scope}\label{stageone.fig} 
\end{tikzpicture}
\caption{The transitions from the single element causal set
  $c_1=\achain_1$ } \label{cone.fig} 
 \end{figure}
 Next, at  stage $n=2$ the transitions  from $\achain_2$ are shown in
 Fig. \ref{antichaintwo.fig} 
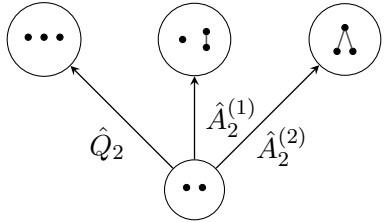
\begin{figure}[H]
 \centering
 \begin{tikzpicture}[-stealth]
  \def\ystep{2cm}
  \def\xstep{0.5cm}
  \begin{scope}[nodes={draw, thin, circle, minimum size=0.5cm}]
 \node (C0) at ( 0, 0) {\pcauset{2,1}};
 \node (C321) at (-4*\xstep, 1*\ystep) {\pcauset{3,2,1}};
 \node (C213) at ( 4*\xstep, 1*\ystep) {\pcauset{2,1,3}};
 \node (C231) at (0,\ystep) {\pcauset{2,3,1}};
 \end{scope}
 \begin{scope}[prob arrow]
  \drawprobleft{C0}{\hQ_2\quad }{C321}
 \drawprobright{C0}{\quad \h A_2^{(2)}}{C213}
  \drawprobrighttop{C0}{\h A_2^{(1)}}{C231}
 \end{scope}
\end{tikzpicture}
\caption{The transitions from $\achain_2$ } \label{antichaintwo.fig} 
 \end{figure}
The transition operators are  obtained using NTOBC and MSR: 
\begin{eqnarray} 
\h A_2^{(1)} = \hQ_1^{-1} \hQ_2 - \hQ_2\nonumber \\
\h A_2^{(2)}= \id - 2 \hQ_1^{-1} \hQ_2 + \hQ_2 \nonumber \\
\end{eqnarray}

Using the expressions for $\hA_2^{(2)}$ and $\hA_2^{(1)}$ in 
Eqn. \eqref{ntobccommrel.eq}, we see that 
\begin{equation}
\bigl[ \hQ_1,\hQ_2\bigr]=0.  \label{onetwoQ.eq} 
\end{equation}

The remaining transitions at   stage $n=2$ from the $2$-element chain,
shown in Fig \ref{chaintwo.fig}.  
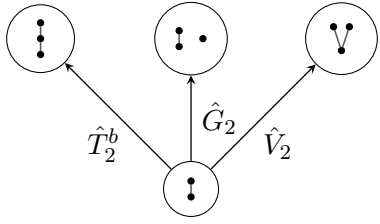
\begin{figure}[H]
 \centering
 \begin{tikzpicture}[-stealth]
  \def\ystep{2cm}
  \def\xstep{0.5cm}
  \begin{scope}[nodes={draw, thin, circle, minimum size=0.5cm}]
 \node (C0) at ( 0, 0) {\pcauset{1,2}};
 \node (C123) at (-4*\xstep, 1*\ystep) {\pcauset{1,2,3}};
 \node (C213) at ( 4*\xstep, 1*\ystep) {\pcauset{1,3,2}};
 \node (C312) at (0,\ystep) {\pcauset{3,1,2}};
 \end{scope}
 \begin{scope}[prob arrow]
  \drawprobleft{C0}{\h T_2^b \quad }{C123}
 \drawprobleft{C0}{\quad \h V_2}{C213}
  \drawprobrighttop{C0}{\h G_2}{C312}
 \end{scope}
\end{tikzpicture}
\caption{The transitions from the $n=2$ chain, $\chain_2$ } 
\label{chaintwo.fig} 
 \end{figure}
The transition operators  are obtained using NTOBC, GC and MSR,
and the relation Eqn. \eqref{onetwoQ.eq} 
\begin{eqnarray}
   \h G_2=Q_2 \nonumber \\ 
   \h V_2 = \hQ_1^{-1} \hQ_2 - \hQ_2 \nonumber \\
   \h T_2^b =\id - \hQ_1^{-1} \hQ_2.
\end{eqnarray}

We now prove inductively  that $\bigl[ \hQ_1,\hQ_n \bigr]=0$ for all
$n$, and therefore that  $\cQ$ has a non-trivial center. Assume that  $\bigl[ \hQ_1,\hQ_{m} \bigr]=0$ for all
$m<n$. Using $ \hA_n^{(1)} \hQ_n^{-1}= \hQ_1^{-1}-\id$ and the form
Eqn. \eqref{ntobctimid.eq} 
for $\hA_n^{(n)}=\h A_n^{(n)} $, 
\begin{eqnarray} 
  [\hA_n^{(n)} \hQ_n^{-1}, \hA_n^{(1)}\hQ_n^{-1}] &= & 0 \nonumber \\
  \Rightarrow \sum_{\ell=0}^n (-1)^\ell \binom{n}{\ell} [(\hQ_{n-\ell})^{-1}
  ,  \hQ_1^{-1}] &= & 0 \nonumber \\
   \Rightarrow [\hQ_n,\hQ_1] &=&0. 
\end{eqnarray} 
Next, using Eqn. \eqref{ntobcnontimid.eq} we note  that 
\begin{equation} 
  [\hA_n^{(3)} \hQ_n^{-1}, \hA_n^{(2)}\hQ_n^{-1}]=0 \Rightarrow
  [\hQ_3,\hQ_2] =0. 
\end{equation}
Assuming  that  $[\hQ_m,\hQ_2] =0$ for all $m<n$, 
\begin{eqnarray} 
  [\hA_n^{(n)} \hQ_n^{-1}, \hA_n^{(2)}\hQ_n^{-1}] &=& 0 \nonumber \\
  \Rightarrow \sum_{\ell=0}^{n} (-1)^\ell \binom{n}{\ell}[\hQ_{n-\ell}^{-1}
  ,  \hQ_{2}^{-1}] &=&0 \nonumber \\
  \Rightarrow [\hQ_n,\hQ_2] &=&0. 
\end{eqnarray}
Assuming that $[\hQ_m,\hQ_k] =0$ for all $k<m<n$, a similar argument
shows that 
\begin{eqnarray} 
  [\hA_n^{(n)} \hQ_n^{-1}, \hA_n^{(k)}\hQ_n^{-1}] &=& 0 \nonumber \\
  \Rightarrow [\hQ_n,\hQ_k] &=&0. 
\end{eqnarray}
Hence $[\hQ_n,\hQ_m]=0$ for all $n,m$.  We have therefore proved that 
\begin{lemma}
The subalgebra $\cQ$ is commutative.  
\end{lemma} 

Next, we show  using induction that not only $\cQ$  but also $\cA$ is commutative.  The transition operators
upto stage  $n=2$, as well as the $\hA_n^{(k)}$ and the $\h A_k^{(k)}$ operators satisfy the 
general form 
\begin{equation}
\hA_n= \sum_{\ell=0}^m (-1)^\ell \binom{m}{l} \hQ_{\varpi - \ell}
^{-1} \hQ_n \label{ansatzntobc.eq}
\end{equation} 
where $\varpi$ is the cardinality of the precursor set, and $m$ is
the number of maximal elements in the precursor set. This is the
ansatz that we will now prove.  

Assume that Eqn. \eqref{ansatzntobc.eq} is true upto stage
 $n-1$. Thus, for all $k<n$,  since $\varpi=0$ for a gregarious
 transition 
 \begin{equation}
\hG_{k}= \h Q_k. 
\end{equation}
In Eqn. \ref{decimation.eq}  since $\tS^{(\alpha)}$ is constructed
purely from the $\h Q_{s}$'s for $s<n$, they commute with  $\hQ_n$ as
well, and hence 
\begin{equation}
\hG_n^{(0)}= \hQ_n. 
\end{equation} 
   Next, for any non-timid transition operator $\hA_n$, let $\hA_s$ be
   its Bell related transition for some $s<n$.  By NOTBC, 
   \begin{equation}
\hA_n=\hA_s \hQ_s^{-1}  \hQ_n= \biggl(\sum_{\ell=0}^m (-1)^\ell \binom{m}{l} \hQ_{\varpi - \ell}
^{-1} \hQ_s \biggr) \hQ_s^{-1}  \hQ_n =  \sum_{\ell=0}^m (-1)^\ell \binom{m}{l} \hQ_{\varpi - \ell}
^{-1} \hQ_n,  
\end{equation} 
which completes our proof.

We have therefore proved that

\begin{lemma} 
The transition operators for NTOBC generate a commutative algebra.  
\end{lemma}

\subsection{Causal Past Ordered (CPOBC)}
\label{cpobc.ssec}

For the related {Bell pairs} $\{\hA_n,\hA_n'\} $ and $\{\hA_m\hA_m' \}
$,  if $\pre(\hA_n)=\pre(\hA_m)$   is strictly larger than
$\pre(\hA_n')=\pre(\hA_m')$, then 
\begin{equation}
\hA_n\hA_m'=\hA_m\hA_n', 
\end{equation}
and when they are equal, i.e.,  $\pre(\hA_n)=\pre(\hA_n')$ 
\begin{eqnarray}
  \hA_n\hA_m' = \hA_m\hA_n' &\quad \& & \hA_n'\hA_m= \hA_m'\hA_n
                                         \nonumber \\ 
 \Rightarrow     \hA_n = \hA_m\hA_n'  \hA_m'^{-1}=  \hA_m'^{-1}
  \hA_n'\hA_m                & \Rightarrow &   \hA_m'\hA_m \hA_n'  =  
  \hA_n'\hA_m \hA_m', 
\end{eqnarray}
which is an additional constraint on the algebra.

Thus, for $\pre(\hA_n)> \pre(\hA_n')$ 
\begin{equation} 
\hhT_n =\tcF(\hhT_n',\hhT_m,\hhT_m') = \hhT_m\hhT_n'  \hhT_m'^{-1},
\end{equation}
while for $\pre(\hA_n)= \pre(\hA_n')$ 
\begin{equation} 
\hhT_n =\tcF(\hhT_n',\hhT_m,\hhT_m') = \hhT_m\hhT_n'  \hhT_m'^{-1}=
\hhT_m'^{-1}\hhT_n'  \hhT_m.
\end{equation}

In particular, for  any non-timid transition,
\begin{equation}
\hA_n=\hA_m\hG_n  \hG_m^{-1}. \label{nontimidcpobc.eq}
  \end{equation} 
The timid transition operator is given by Eqn. \eqref{timid.eq} (Lemma \ref{GTA.lemma}) 
\begin{eqnarray} 
  \hTi_n&=&\id-\hG_n-\sum_{\ell=1}^\bmm \sum_{k_1<\ldots k_\ell}^\bmm
            (-1)^{\ell -1}
            \biggl(\id -\hG_{n-\ell}^{k_1\ldots
            k_\ell}\biggr)\hG_n\biggl(\hG_{n-\ell}^{k_1\ldots
            k_\ell}\biggr)^{-1}  \nonumber \\
  &=& \sum_{\ell=0}^\bmm \sum_{k_1<\ldots
      k_\ell}^\bmm  (-1)^{\ell} \biggl(\id - \hG_{n-\ell}^{k_1\ldots
            k_\ell}\biggr) \hG_n\biggl(\hG_{n-\ell}^{k_1\ldots
            k_\ell}\biggr)^{-1}.  \label{timidcpobc.eq} 
\end{eqnarray} 

As before we can use the atomisation process  Fig \ref{atomisation.fig} to relate
gregarious transitions at stage $n$.       
Using GC,  
\begin{eqnarray}
\h G_{n}^{(0)} \h B_{n-1}^{(0)}= \h B_{n}^{(1)} \h G_{n-1}^{(0)} 
\end{eqnarray}
while from CPOBC
\begin{equation}
\h B_{n}^{(1)} \h G_{n-1}^{(0)} = \h B_{n-1}^{(0)} \h G_{n}^{(1)}  
\end{equation}
From these, we get that, 
\begin{equation}\label{CPOBCconj.eq}
\h G_{n}^{(0)}=\h B_{n-1}^{(0)}\h G_{n}^{(1)}\h
(B_{n-1}^{(0)})^{-1}. 
\end{equation}
Using an  atomisation path $\alpha$, therefore 
\begin{equation}
  \h G_{n}^{(0)}=\h B_{n-1}^{(0)}\h B_{n-1}^{(1)} \ldots \h
  B_{n-1}^{(k)}\h Q_{n} (\h B_{n-1}^{(k)})^{-1} \ldots (\h B_{n-1}^{(1)})^{-1} \h
  (B_{n-1}^{(0)})^{-1} \equiv \h S^{\alpha} \h Q_{n} (\h
  S^{\alpha})^{-1}. \label{gregcpobc.eq} 
\end{equation}
Since the atomisation path is non-unique, for distinct paths $\alpha$,
$\beta$ we have the commutation relation 
\begin{equation} 
\Rightarrow [\h S_\alpha^{-1}\h S_\beta,\h Q_{n+1}]=
0. \label{salphasbeta.eq} 
\end{equation}

The difference between this case and the previous cases is that we haven't proven that the algebra is commutative. In the following section, we will derive some of the constraints on the dynamics in CPOBC which will help us verify if some choice of $Q_n$s is consistent or not. 

We begin by proving the following theorem that relates gregarious transitions at different stages. 
\begin{theorem}
Suppose that $\h G_n$, $\h G_k$, $\h G_m$ are gregarious transition operators corresponding to gregarious transition from causal sets $c_n$, $c_k$ and $c_m$ respectively, where $m\neq n$ and $k<\min \{m,n\}$. Then they satisfy the following relation
\begin{equation}\label{GTC3.eq}
\h G_n\h G_k^{-1}\h G_m=\h G_m\h G_k^{-1}\h G_n
\end{equation}
\end{theorem}
\begin{proof}
Consider transitions from $c_n$ with a single 
element in the precursor set, which is therefore a minimal element in
$c_n$.  Such a transitions  are therefore all  Bell related to each
other, and in particular to the transition from the
single element causal set to the $2$-element chain.  Let 
the corresponding transition operators from $c_n$ and $c_m$ be $\h
B_n$ and $\h B_m$ respectively. Then by CPOBC 
 \begin{equation}
 \h B_n=\h B_m\h G_n\h G_m^{-1}
 \end{equation}
 Similarly, $\h B_k,\h B_m, \h B_n$ are related as 
 \begin{eqnarray}
\h B_m&=&\h B_k\h G_m\h G_k^{-1}\\
 \h B_n&=&\h B_k\h G_n\h G_k^{-1}
 \end{eqnarray}
 From these equations, we get,
\begin{eqnarray}
\h B_k\h G_n\h G_k^{-1}&=&\h B_k\h G_m\h G_k^{-1}\h G_n\h G_m^{-1}\\
\h G_n\h G_k^{-1}\h G_m&=&\h G_m\h G_k^{-1}\h G_n \label{gregrels.eq}
\end{eqnarray}
\end{proof}
In particular, we have
\begin{equation}
\h Q_n\h Q_k^{-1}\h Q_m=\h Q_m\h Q_k^{-1}\h Q_n.
\end{equation}

We also have the following relation between gregarious transitions which can be derived from this theorem.

\begin{corollary}
If we have $n,m,k,\ell\in \mathbb{N}$ such that $m\neq n \neq k \neq \ell $, then the gregarious transition operators from $c_n, c_m,c_k$ and $c_\ell$ are related by the following equation
\begin{equation}\label{4Greg3.eq}
[\h G_m^{-1}\h G_n,\h G_k^{-1}\h G_\ell]=0 
\end{equation}
\end{corollary}
\begin{proof}
Lets first assume that $k<\ell ,m,n$ and use equation \eqref{GTC3.eq} to relate the gregarious transition operators as 
\begin{eqnarray}
\h G_m\h G_k^{-1}\h G_\ell &=&\h G_\ell\h G_k^{-1}\h G_m\nonumber\\
\h G_m\h G_k^{-1}\h G_n&=&\h G_n\h G_k^{-1}\h G_m. 
\end{eqnarray}
Combining, we find that 
\begin{eqnarray}
\h G_k^{-1}\h G_m\h G_\ell^{-1}\h G_n&=& \h G_\ell^{-1}\h G_n\h G_k^{-1}\h G_m\\
\Rightarrow [\h G_k^{-1}\h G_m,\h G_\ell^{-1}\h G_n]&=&0
\end{eqnarray}
We can take the inverse of above equation to get
\begin{equation}
[\h G_n^{-1}\h G_\ell ,\h G_m^{-1}\h G_k]=0
\end{equation}
We also note that for operators $\h A ,\h B$, $[\h A,\h B]=0\Rightarrow [\h A^{-1},\h B]=0$. We thus have
\begin{eqnarray}
\left[\h G_\ell^{-1}\h G_n ,\h G_m^{-1}\h G_k\right]&=&0\\
\left[\h G_n^{-1}\h G_\ell ,\h G_k^{-1}\h G_m\right]&=&0
\end{eqnarray}
With that, we see that $k$ can be any of the four index in equation \eqref{4Greg3.eq} and so is true for $n,m,\ell$. Thus, we showed that for any $n,m,k,\ell\in \mathbb{N}$ such that $m\neq n \neq k \neq \ell $
\begin{equation}
[\h G_m^{-1}\h G_n,\h G_k^{-1}\h G_\ell]=0 
\end{equation}
\end{proof}

In particular, this relation implies 
\begin{equation}\label{eq: constraint}
[\h Q_m\h Q_n^{-1},\h Q_\ell \h Q_k^{-1}]=0
\end{equation}
The equation \eqref{GTC3.eq} gives us one more relation between $\h Q_1$ and $\h Q_2$ which further constraints the dynamics.

\begin{corollary}
$\h Q_1$ and $\h Q_2$ are constrained by the equation
\begin{equation}
[\h Q_1\h Q_2^{-1},\h Q_1^{-1}\h Q_2]=0
\end{equation}
\end{corollary}
\begin{proof}
Using $k=1,n=2,m=3$  in Eqn \eqref{GTC3.eq} 
\begin{equation}
\h G_2\h Q_1^{-1}\h G_3=\h G_3\h Q_1^{-1}\h G_2
\end{equation}
Consider  the following transitions: 
\begin{center}
\begin{tikzpicture}[-stealth]
  \def\ystep{1.5 cm}
  \def\xstep{0.5 cm}
  \begin{scope}[nodes={draw, thin, circle, minimum size=0.5cm}]
 \node (C1) at ( 0, 0) {\pcauset{1}};
 \node (C12) at (-4*\xstep, 1*\ystep) {\pcauset{1,2}};
 \node (C21) at ( 4*\xstep, 1*\ystep) {\pcauset{2,1}};
 \node (C312) at ( -4*\xstep, 2*\ystep) {\pcauset{3,1,2}};

 \node (C321) at (4*\xstep , 2*\ystep) {\pcauset{3,2,1}};
 
 \node (C4312) at (-4*\xstep , 3*\ystep) {\pcauset{4,3,1,2}};
 \node (C4321) at (4*\xstep , 3*\ystep) {\pcauset{4,3,2,1}};
  \end{scope} 
\begin{scope}[prob arrow,black]
  \drawprobarrow{C1}{\h A_1^{(1)}}{C12}
  \drawprobarrow{C1}{\h Q_1}{C21}
  \drawproblefttop{C12}{\h G_2}{C312}
  \drawprobrighttop{C21}{\h Q_2}{C321}
   \drawprobarrow{C21}{\h A_2^{(1)}}{C312}
    \drawprobarrow{C321}{\h A_3^{(1)}}{C4312}
        \drawprobrighttop{C321}{\h Q_3}{C4321}
            \drawproblefttop{C312}{\h G_3}{C4312}
  \end{scope} 
 \end{tikzpicture}
\end{center} 
From Eqn \eqref{CPOBCconj.eq} 
\begin{eqnarray}
\h G_2&=&\h B_1 \h Q_2\h B_1^{-1}     \nonumber \\
\h G_3&=& \h B_2\h Q_3\h B_2^{-1}     \nonumber \\
&=& \h B_1\h Q_2\h Q_1^{-1}\h Q_3\h Q_1\h Q_2^{-1}\h B_1^{-1}
\end{eqnarray}
Putting these together, we get
\begin{eqnarray}
 \h B_1\h Q_2\h Q_1^{-1}\h Q_3\h Q_1\h Q_2^{-1}\h B_1^{-1}\h Q_1^{-1}\h B_1 \h Q_2\h B_1^{-1}&=&\h B_1 \h Q_2\h B_1^{-1}\h Q_1^{-1}\h B_1\h Q_2\h Q_1^{-1}\h Q_3\h Q_1\h Q_2^{-1}\h B_1^{-1}\nonumber \\
 \Rightarrow \h Q_3\h Q_1\h Q_2^{-1}\h Q_1^{-1}\h Q_2&=&\h Q_2\h
                                                         Q_1^{-1}\h
                                                         Q_3\hQ_1\hQ_2^{-1}. 
\end{eqnarray}
Next, using $\h Q_3\h Q_1^{-1}\h Q_2=\h Q_2\h Q_1^{-1}\hQ_3$
(Eqn. \eqref{gregrels.eq} ), we get  
\begin{eqnarray} 
\hQ_1\hQ_2^{-1}\hQ_1^{-1}\hQ_2 &=& \hQ_1^{-1}\hQ_2\hQ_1\hQ_2^{-1}
                                   \nonumber \\ 
\Rightarrow [\h Q_1\h Q_2^{-1},\h Q_1^{-1}\h Q_2]&=&0 \label{qoneqtwo.eq}
\end{eqnarray} 
\end{proof}

While there are several of these identities,  none of these imply that
the $\hQ_i$ themselves commute.  

As before let us begin with the transition operators from the
$n$-antichain $\achain_n$.  For the timid transition, 
\begin{equation}
\h A_n^{(n)}= \sum_{\ell=0}^n (-1)^{\ell} \binom{n}{\ell}
\biggl(\id - \hQ_{n-\ell}\biggr) \hQ_n \hQ_{n-\ell}^{-1}.   \label{timidachaincbc.eq}
\end{equation} 
Using CPOBC, moreover, the non-timid transition operator 
$\hA_n^{(k)}$  from $\achain_n$ with $k$ precursor elements is Bell
related to $\h A_k^{(k)}$. Using 
\begin{equation}
\hA_n^{(k)}\hQ_k=\h A_k^{(k)}\hQ_n 
\end{equation} 
and Eqn. \eqref{timidachaincbc.eq}
\begin{eqnarray} 
\hA_n^{(k)} & = &  \sum_{\ell=0}^{k} (-1)^{\ell} \binom{k}{\ell}
\biggl(\id - \hQ_{k-\ell}\biggr) \hQ_k \hQ_{k-\ell}^{-1} \hQ_n
                  \hQ_k^{-1}  \nonumber \\
&=& \sum_{\ell=0}^{k} (-1)^{\ell} \binom{k}{\ell}
\biggl(\id - \hQ_{k-\ell}\biggr) \hQ_n
    \hQ_{k-\ell}^{-1} \label{ankcpobc.eq} 
\end{eqnarray} 
where we have used
\begin{equation}
\hQ_k \hQ_{k-\ell}^{-1} \hQ_n = \hQ_n\hQ_{k-\ell}^{-1} \hQ_k. 
  \end{equation} 

Following a similar set of arguments from
Eqn. \eqref{ntobccommrel.eq}-\eqref{ntobcone.eq}, however, we get the
rather obtuse expression for $m,k<n$ 
\begin{equation}
\hA_n^{(m)}\hA_n^{(k)} \hQ_{n+1} \biggl(\hA_n^{(k)}\biggr)^{-1}
\hQ_n^{-1} \hA_n^{(k)} =\hA_n^{(k)}\hA_n^{(m)} \hQ_{n+1} \biggl(\hA_n^{(m)}\biggr)^{-1}
\hQ_n^{-1} \hA_n^{(m)}.  \label{cpobc.eq}
\end{equation}

Next, we calculate the transition operators  upto
$n=2$ in order to try to find a general expression for the 
transition operators. 
For the transitions from the antichain $\achain_2$ shown in
Fig. \ref{antichaintwo.fig}, we find  using CPOBC for the   Bell related pairs $\{\hA_2^{(1)},
\hQ_2 \}$ and $\{ (\id -\hQ_1), \hQ_1\}$ and subsequently the MSR,
that 
\begin{eqnarray}
\h A_2^{(1)}&=&(\mathbb{I}-\h Q_1) \h Q_2\h Q_1^{-1}=\h Q_2\h Q_1^{-1}-\h Q_1\h Q_2\h Q_1^{-1}\\
\h A_2^{(2)}&=&\mathbb{I}-2\h B_2-\h Q_2 =\mathbb{I}-2 (\mathbb{I}-\h Q_1) \h Q_2\h Q_1^{-1}-\h Q_2. 
\end{eqnarray}
The transition operators from the $2$-element chain  
are shown in Fig. \ref{chaintwo.fig}.  Using GC, to relate $\hA_2^{(1)}$ with $\hG_2$,  and the Bell related
 pairs $\{ \h V_2, \hG_2 \}, \{ \id - \hQ_1, \hQ_1\}$, as well as the
 MSR, we find that   
 \begin{eqnarray}
 \h G_2&=&(\mathbb{I}-\h Q_1)\h Q_2(\mathbb{I}-\h Q_1)^{-1}\\
 \h V_2&=&(\mathbb{I}-\h Q_1)\h G_2\h Q_1^{-1}=(\mathbb{I}-\h Q_1)^2\h Q_2 (\mathbb{I}-\h Q_1)^{-1}\h Q_1^{-1}\\
 \h T_2^b&=&\mathbb{I}-\h X_2-\h G_2 = \mathbb{I}- (\mathbb{I}-\h
             Q_1)\h Q_2(\mathbb{I}-\h Q_1)^{-1}-(\mathbb{I}-\h
             Q_1)^2\h Q_2 (\mathbb{I}-\h Q_1)^{-1}\h Q_1^{-1} 
 \end{eqnarray}
 Note that if  $[ \hQ_1, \hQ_2] \neq 0$ then $\h G_2 \neq \h
 Q_2$. 

 The expressions for $\hA_2^{(1)}$ and $\h A_2^{(2)}$ can be inserted into Eqn. \eqref{cpobc.eq} for
$n=2$ to give
\begin{eqnarray} 
&& (\mathbb{I}-\h Q_1) \h Q_2\h Q_1^{-1} \biggl(\mathbb{I}-2
  (\mathbb{I}-\h Q_1) \h Q_2\h Q_1^{-1}-\h Q_2\biggr) \h Q_3  \biggl(\mathbb{I}-2
  (\mathbb{I}-\h Q_1) \h Q_2\h Q_1^{-1}-\h Q_2\biggr)^{-1} \nonumber
  \\ &\times & \hQ_2^{-1} \biggl(\mathbb{I}-2
  (\mathbb{I}-\h Q_1) \h Q_2\h Q_1^{-1}-\h Q_2\biggr)   \nonumber \\
  &=&  \biggl(\mathbb{I}-2
  (\mathbb{I}-\h Q_1) \h Q_2\h Q_1^{-1}-\h Q_2\biggr) (\mathbb{I}-\h
      Q_1) \h Q_2\h Q_1^{-1} \h Q_3 \h Q_1 \h Q_2^{-1} \nonumber \\ &\times &  (\mathbb{I}-\h Q_1)^{-1} \h Q_2^{-1}
       (\mathbb{I}-\h Q_1) \h Q_2\h Q_1^{-1}, 
  \end{eqnarray} 
which is a fairly elaborate relation between $\hQ_1,\hQ_2$ and
$\hQ_3$, and does not simplify significantly. Thus, despite the many
new relations, we do not find a general form for the transition
operators. 

We note that 
\begin{lemma}
Assume that  $\exists \, \, k$ such that $\forall \, \,  n$, $[\h Q_n,\h
Q_k]=0$, namely that $\hQ_k$ belongs to the center of $\cQ$. Then 
$\cQ$ is commutative, as is  $\cA$. 
\end{lemma}
\begin{proof}
Assume that there exists a $k\in \mathbb{N}$ exists such that $[Q_n,Q_k]=0$
$\forall n\in \mathbb{N}$. If $n<\min \{k,m\}$, from Eqn \ref{GTC3.eq} 
\begin{eqnarray}
\h Q_k\h Q_n^{-1}\h Q_m&=&\hQ_m\hQ_n^{-1}\hQ_k               \nonumber \\
\Rightarrow \hQ_n^{-1}\hQ_m&=&\hQ_m\hQ_n^{-1}  \nonumber \\
\Rightarrow [\hQ_n,\hQ_m]&=&0
\end{eqnarray}
Also, if $k<\min{n,m}$, 
\begin{eqnarray}
\hQ_n\hQ_k^{-1}\hQ_m&=&\hQ_m\hQ_k^{-1}\hQ_n   \nonumber\\
\Rightarrow \hQ_n\hQ_m&=&\hQ_m\hQ_n  \nonumber  \\
\Rightarrow [\hQ_n,\hQ_m]&=&0. 
\end{eqnarray}
This covers all possibilities and hence we see that  $[\hQ_n,\hQ_m]=0$
for all $n,m$.

From Eqn. \eqref{nontimidcpobc.eq} every non-timid transition operator is given
in terms of the gregarious transition at stage $n$ as well as a 
corresponding timid transition at stage $m$,
\begin{equation}
  \hA_n=\hT_m\hG_n \hG_m^{-1}
\end{equation} 
In turn,  $\hT_m$ is given in terms of $\hG_k$'s, $k \leq m$  (Eqn. \eqref{timidcpobc.eq})
and the $\hG_n$'s themselves are conjugate to $\hQ_n$ via a string of
$(n-1)$ stage operators, Eqn. \eqref{gregcpobc.eq}.  
 
At stage $n=2$ we see that for commuting $\hQ_i$,
$\hG_2=\hQ_2$ and moreover, that all the transition operators simplify
to that of NTOBC.  Hence we can again assume the general form of the
transition operator as we did for NTOBC and show, using induction,
that the algebra is commutative.

\end{proof}
In  Appendix \ref{appendix.app} explore this algebra further, going
upto stage $n=3$, but again, there is no general form that we
find. Instead, we use a $d=2$
representation of the $\h Q_i$'s  using Pauli matrices $\sigma_i$, to
check if they generate the algebra. We find that given the assumption
of invertibility for each of the transition operators, this choice of
$\h Q_i$'s does not generate the CPOBC algebra.

\section{Discussion}
\label{discussion.sec}

The decoherence functional for two cylinder sets is given by 
\begin{equation} 
  D(\cyl{c_n},\cyl{c_m})=\bra{\Omega}\h A_1^\ast \h A_2^\ast \ldots \h
  A^\ast_{n-1} \h B_{m-1} \ldots \h B_2 \h B_1
  \ket{\Omega} \equiv \bra{\Omega} \h D(\gamma(c_n),\gamma(c_m))\ket{\Omega}
\end{equation}
where $\gamma(c_n)$ denotes the path in $\cP$ from $c_1$ to $c_n$. In unitary quantum mechanics, the operators $\h A_k$ are 
projection operators $\hA_k=\ket{x_k}\bra{x_k}$, and hence for two
paths $\gamma_1, \gamma_2$,  at the final time step, $T=n-1$,
${A^\ast}_{T}B_{T}= \braket{\gamma_1(T)}{\gamma_2(T)} =\delta(\gamma_1(T)-\gamma_2(T))$. If we are
to mimic such a behaviour in QSG, then it is clear that the commutative $\CSG$ models
will not suffice, since the decoherence functional reduces to the
product form  $\h D(\gamma(c_n),\gamma(c_m)) = \h
A(\gamma(c_n))^\ast \h A(\gamma(c_m))$. Hence one is led to the
larger class of  non-commutative QSG models. 

However, workable non-abelian QSG models seem elusive. 
In this work we have explored various implementations of the QBC and
find that the two most natural choices lead to a commutative
algebra. In the third choice, we find new commutation
relations, but because of the complexity of the algebra, are unable to
find a closed form expression for the general transition
operator. Nevertheless, we show by example 
that the constraints on this algebra are stringent enough to rule out
a Pauli matrix algebra representation. 

Despite these difficulties, we view this work as a first step towards
finding calculable,
non-abelian QSG
models. If such a class of models could be found, then  the next
step would be  to determine if the quantum measure on $\fA$ extends to
its sigma
algebra completion $\fS$, as  required by the
Caratheodary-Hahn-Kluvnek theorem \cite{diestel,sz}. This would give
us a class of fully covariant non-abelian QSG models, from which one
would then need to tease out realistic models.

Apart from operator ordering ambiguities in implementing QBC, one could stipulate that it
isn't the quantum vector measure that should be the same when
implementing GC, but rather only its norm. Thus, one could introduce a
path dependent phase into the definition. Such phases can similarly be
introduced into the QBC relations.  While we have explored these
options, apart from trivial choices, the ensuing dynamics becomes
rapidly more complex. Moreover, if the transition operators are allowed to be singular, then only two
of the three types  of QBC  can be used. In these cases, however, the  resulting dynamics
generates newer unconstrained  parameters at every stage, signalling a
form of unpredictability. We leave further explorations of these
generalisations to the future.

We end by speculating on the possibility that 
there may be  other ways to obtain QSG models,  without resorting to
local rules on the algebra as we have done in this work, but rather
invoking more global algebraic structures as in AQFT.

\section{Acknowledgements} 
SS was supported in part by the ANRF MATRICS grant,
MTR/2023/000831 as well as the WISER Indo-German grant,  IGSTC/WISER
25/P124-SS5876-AE5909/64/2025-26. 

\appendix
\section*{Appendix}
\label{appendix.app}

In this appendix we explore the CPOBC algebra further. We find $n=3$ and $n=4$ transitions which provide a non-trivial
example of the relation Eqn. \eqref{salphasbeta.eq}. This gives us 
two different paths $S_1, S_2$ corresponding to different
atomisations  of the same causal
set. By themselves these do not lead to useful simplifications of the
algebra. Rather than look at the full abstract algebra for more
relations, we attempt to find a simple non-trivial representation
of $\cA$ starting from a representation of $\cQ$. We find that while
the $\cQ$ can be represented as a Pauli matrix algebra, it is no
longer a subalgebra of $\cA$.

We begin by calculating the transition operators from (a)
$\achain_3$: 
\begin{figure}[H]
\centering
 \begin{tikzpicture}[-stealth]
  \def\ystep{2.6cm}
  \def\xstep{0.5cm}
  \begin{scope}[nodes={draw, thin, circle, minimum size=0.5cm}]
 \node (C0) at ( 0, 0) {\pcauset{3,2,1}};
 \node (C3421) at (-4*\xstep, 1*\ystep) {\pcauset{3,4,2,1}};
 \node (C3241) at ( 4*\xstep, 1*\ystep) {\pcauset{3,2,4,1}};
 \node (C3214) at (-8*\xstep,\ystep) {\pcauset{3,2,1,4}};
 \node (C4321) at (8*\xstep,\ystep) {\pcauset{4,3,2,1}};

 \end{scope}
\begin{scope}[prob arrow]
  \drawprobrighttop{C0}{\,\h A_3^{(1)}}{C3421}
 \drawproblefttop{C0}{\h A_3^{(2)}}{C3241}
  \drawprobright{C0}{\quad\hQ_3}{C4321}
  \drawprobleft{C0}{\h C_3\quad }{C3214}
 
  \end{scope}
 \end{tikzpicture}
\end{figure}
which are (Eqn. \eqref{ankcpobc.eq}) 
\begin{eqnarray}
\h A_3^{(1)}&=&\h A_2^{(1)}\hQ_3\hQ_2^{-1}\\
\h A_3^{(2)}&=&\h A_2^{(2)}\hQ_3\hQ_2^{-1}\\
\h A_3^{(3)}&=& \mathbb{I}-3\h A_3^{(1)}-3\h A_3^{(2)}-\hQ_3. 
\end{eqnarray}
and from  (b) $\pcauset{3,1,2}$:  
\begin{figure}[H]
\centering
 \begin{tikzpicture}[-stealth]
  \def\ystep{2.6cm}
  \def\xstep{0.75cm}
  \begin{scope}[nodes={draw, thin, circle, minimum size=0.5cm}]
 \node (C0) at ( 0, 0) {\pcauset{2,3,1}};
 \node (C2314) at (-7.5*\xstep, 1*\ystep) {\pcauset{2,3,1,4}};
 \node (C2413) at (-4.5 *\xstep, 1*\ystep) {\pcauset{2,4,1,3}};
 \node (C2341) at (-1.5*\xstep,\ystep) {\pcauset{2,3,4,1}};
 \node (C3412) at (1.5*\xstep,\ystep) {\pcauset{3,4,1,2}};
 \node (C2431) at (4.5*\xstep,\ystep) {\pcauset{2,4,3,1}};
 \node (C3421) at (7.5*\xstep,\ystep) {\pcauset{3,4,2,1}};

 \end{scope}
\begin{scope}[prob arrow]
  \drawprobleft{C0}{\h H_3}{C2314}
 \drawproblefttop{C0}{\h\Lambda_3^b\,}{C2413}
 \drawproblefttop{C0}{\h F_3}{C2341}
 \drawprobrighttop{C0}{\h E_3}{C3412}
  \drawprobrighttop{C0}{\h E_3'}{C2431}
   \drawprobright{C0}{\h D_3}{C3421}
    \end{scope} 
 \end{tikzpicture}
\end{figure}
\begin{eqnarray}
\h D_3&=&\h A_2^{(1)}\hQ_3(\h A_2^{(1)})^{-1}\\
\h E_3&=&\h E_3'=(\mathbb{I}-\hQ_1)\h D_3\hQ_1^{-1}\\
\h F_3&=&\h T_2^b\h D_3(\mathbb{I}-\hQ_1)\hQ_2^{-1}(\mathbb{I}-\hQ_1)^{-1}\\
\h\Lambda_3^b&=& \h A_2^{(2)} \h D_3\hQ_2^{-1}\\
\h H_3&=&\mathbb{I}-\h\Lambda_3^b-\h F_3-2\h E_3-\h D_3, 
\end{eqnarray}
where we have used MSR, GC and CPOBC.

Next, consider the causal set  \pcauset{2,4,1,3} and the gregarious
transition $\h G_4:  \pcauset{2,4,1,3} \rightarrow \pcauset{5,2,4,1,3}
$. There  two possible atomisation-decimation diagrams associated with
\pcauset{2,4,1,3}. The first is:  
\begin{figure}[H]
\centering
 \begin{tikzpicture}[-stealth]
  \def\ystep{1.5cm}
  \def\xstep{0.5cm}
  \begin{scope}[nodes={draw, thin, circle, minimum size=0.5cm}]
\node (C321) at (2,2) {\pcauset{3,2,1}};
    \node (C21) at (0,0) {\pcauset{2,1}};
\node (C213) at (-2,2) {\pcauset{2,1,3}};
 \node (C2413) at (-4,4) {\pcauset{2,4,1,3}};
 \node (C4213) at ( 0,4) {\pcauset{4,2,1,3}};
 \node (C52413) at (-2,6) {\pcauset{5,2,4,1,3}};
 \node (C54213) at (2,6) {\pcauset{5,4,2,1,3}};
 \node (C4321) at (4,4){\pcauset{4,3,2,1}};
 \node (C54321) at (6,6){\pcauset{5,4,3,2,1}};
 \end{scope}
\begin{scope}[prob arrow]
\drawprobarrow{C21}{\h Q_2}{C321}

\drawprobarrow{C21}{\h A_2^{(2)}}{C213}
\drawprobarrow{C321}{\h A_3^{(2)}}{C4213}
  \drawprobarrow{C4213}{\h B_4}{C52413}
 \drawprobarrow{C213}{\h G_3}{C4213}
  \drawprobarrow{C2413}{\h G_4}{C52413}
  \drawprobarrow{C213}{\hat{B}_3}{C2413}
 \drawprobarrow{C321}{\hQ_3}{C4321}
 \drawprobarrow{C4321}{\h Q_4}{C54321}
 \drawprobarrow{C4213}{\h G_4'}{C54213}
 \drawprobarrow{C4321}{\h A_4^{(2)}}{C54213}
  \end{scope}
 \end{tikzpicture}
\end{figure}
From this, we see that for this atomisation, $\h G_4=\h S_1\hQ_4\h S_1^{-1}$ with 
\begin{equation} 
\h S_1=\h B_3\h A_3^{(2)}.
\end{equation}
 Using $\h G_3= \h A_2^{(2)}\hQ_3(\h A_2^{(2)})^{-1}$ and the Bell
 related pairs $\{\h B_3, \h G_3\}$ and  $\{\id -\hQ_1, \h Q_1\}$, 
\begin{equation}
  \h S_1= (\id -\h Q_1 ) \h A_2^{(2)}\hQ_3(\h
  A_2^{(2)})^{-1}\hQ_1^{-1}\h A_3^{(2)}.
\end{equation} 
The second atomisation-decimation diagram of \pcauset{2,4,1,3} is: 
\begin{figure}[H]
\centering
 \begin{tikzpicture}[-stealth]
  \def\ystep{1.5cm}
  \def\xstep{0.5cm}
  \begin{scope}[nodes={draw, thin, circle, minimum size=0.5cm}]
\node (C321) at (2,2) {\pcauset{3,2,1}};
    \node (C21) at (0,0) {\pcauset{2,1}};
\node (C213) at (-2,2) {\pcauset{2,3,1}};
 \node (C2413) at (-4,4) {\pcauset{2,4,1,3}};
 \node (C4213) at ( 0,4) {\pcauset{4,3,1,2}};
 \node (C52413) at (-2,6) {\pcauset{5,2,4,1,3}};
 \node (C54213) at (2,6) {\pcauset{5,4,3,1,2}};
 \node (C4321) at (4,4){\pcauset{4,3,2,1}};
 \node (C54321) at (6,6){\pcauset{5,4,3,2,1}};
 \end{scope}
\begin{scope}[prob arrow]
\drawprobarrow{C21}{\h Q_2}{C321}

\drawprobarrow{C21}{\h A_2^{(1)}}{C213}
\drawprobarrow{C321}{\h A_3^{(1)}}{C4213}
  \drawprobarrow{C4213}{\h \Lambda_4^b}{C52413}
 \drawprobarrow{C213}{\h G_3'}{C4213}
  \drawprobarrow{C2413}{\h{ G_4}}{C52413}
  \drawprobarrow{C213}{\h \Lambda_3^b}{C2413}
 \drawprobarrow{C321}{\hQ_3}{C4321}
 \drawprobarrow{C4321}{\h Q_4}{C54321}
 \drawprobarrow{C4213}{\h {\tilde{ G_4'}}}{C54213}
 \drawprobarrow{C4321}{\h A_4^{(1)}}{C54213}
  \end{scope}
 \end{tikzpicture}
\end{figure}
From this we see that  $\h{G_4}=\h S_2 \h Q_4\h S_2^{-1}$ where
\begin{equation}
  \h S_2=\h  \Lambda_3^b\h A_3^{(1)}.
\end{equation}
Using $\h G'_3=\hA_2^{(1)}\hQ_3 (\hA_2^{(1)})^{-1}$,
and the Bell related pairs $\{ \h  \Lambda_3^b, \h G'_3\}$ and $\{ \h
A_2^{(2)}, \h Q_2\}$ 
\begin{equation} 
\h \Lambda_3^b= \h A_2^{(2)}\h G_3'\h Q_2^{-1} = \h
A_2^{(2)}\h A_2^{(1)}\h Q_3 (\hA_2^{(1)})^{-1}\h Q_2^{-1}. 
\end{equation}
Hence
\begin{equation} 
\h S_2=  \hA_2^{(2)}\hA_2^{(1)}\h Q_3 (\hA_2^{(1)})^{-1}\h Q_2^{-1} \h
A_2^{(1)}\hQ_3\hQ_2^{-1}.
\end{equation} 
Eqn. \eqref{salphasbeta.eq} then implies that
\begin{equation} [\h S_2^{-1}\h
  S_1,\hQ_4]=0. \label{sonestwo.eq} 
  \end{equation} 
While this is an explicit  realisation of Eqn. \eqref{salphasbeta.eq},
however, it doesn't by itself provide more information about the
algebra.

Instead we will attempt to construct a representation of $\cA$
starting from a simple representation of the subalgebra $\cQ$ in terms
of Pauli matrices, $\hQ_n\propto \sigma_i$.
We first note that it is not possible for $\hQ_n=a_n\sigma_1$,
$\hQ_m=a_m\sigma_2$ and $\h Q_k=a_k \sigma_3$ for any choice of
$n,m,k\in\mathbb{N}$, since 
\begin{eqnarray}
\hQ_n\hQ_k^{-1}\hQ_m&=&\frac{a_na_m}{a_k}\sigma_1\sigma_3\sigma_2\nonumber\\
&=&-\frac{a_na_m}{a_k}\sigma_2\sigma_3\sigma_1\nonumber\\
\hQ_n\hQ_k^{-1}\hQ_m&=&-\hQ_m\hQ_k^{-1}\hQ_n
\end{eqnarray}
which violates Eqn \ref{gregrels.eq}. Thus, we must  choose only two
of the $\sigma$ matrices for the $\hQ_n$'s. Without loss of
generality, we choose those to be $\sigma_1$ and $\sigma_2$.

In light of the relation Eqn. \eqref{sonestwo.eq} we  begin with the choice 
\begin{equation}
  \hQ_1=a_1\sigma_1,  \, \hQ_2=a_2\sigma_2, \, \hQ_3=a_3\sigma_1 \,
  \hQ_4=a_4\sigma_2. 
\end{equation}
This gives us \footnote{The following
  computations were aided by the use of the Pauli Algebra package in
  Mathematica \cite{PAM} (\url{https://github.com/EverettYou/PauliAlgebra.})}
\begin{equation}
\h S_1 =\frac{\left(a_1^2+3\right) a_3^2 \sigma_2}{(a_1-1) (a_1+1) a_2}-\frac{i \left(3 a_1^2+1\right) a_3^2 \sigma_3}{(a_1-1) a_1 (a_1+1) a_2}-\frac{\left(7 a_1^2-3\right) a_3^2 \sigma_1}{(a_1-1) a_1 (a_1+1)}-\frac{\left(3 a_1^4+3 a_1^2-2\right) a_3^2 \mathbb{I}}{(a_1-1) a_1^2 (a_1+1)}
\end{equation}
and 
\begin{eqnarray}
\h S_2&=& \frac{a_3^2 \sigma_1\left(45 a_1^2 a_2^2+3 a_1^2-20 a_2^2\right)}{a_1 \left(9 a_1^2 a_2^2-a_1^2-4 a_2^2\right)}+\frac{a_3^2 \sigma_2\left(9 a_1^2 a_2^2-a_1^2+36 a_2^2\right)}{a_2\left(9 a_1^2 a_2^2-a_1^2-4 a_2^2\right)}\nonumber\\
&&-\frac{i a_3^2 \sigma_3 \left(33 a_1^2 a_2^2-a_1^2+12 a_2^2\right)}{a_1 a_2 \left(9 a_1^2 a_2^2-a_1^2-4 a_2^2\right)}-\frac{a_3^2 \mathbb{I} \left(27 a_1^4 a_2^2-3 a_1^4+6 a_1^2 a_2^2+6 a_1^2-8 a_2^2\right)}{a_1^2 \left(9 a_1^2 a_2^2-a_1^2-4 a_2^2\right)}
\end{eqnarray}
Since these operators are required to be invertible, 
\begin{eqnarray}
\det \h S_1&=&\frac{(a_1-1)^2 (a_1+1)^2 a_3^8 \left(9 a_1^2
               a_2^2-a_1^2-4 a_2^2\right)^2}{a_1^8 a_2^4} \neq 0 \nonumber\\
\det \h S_2&=&\frac{(a_1-1)^2 (a_1)^2 a_3^8 \left(9 a_1^2
               a_2^2-a_1^2-4 a_2^2\right)^2}{a_1^8 a_2^4} \neq 0. 
\end{eqnarray}
Next, using 
\begin{eqnarray}
\h  S_1^{-1}\h S_2&=&-\frac{2 a_1^2 a_1 \sigma_2\left(207 a_1^4 a_2^2+a_1^4+25 a_1^2 a_2^2-5 a_1^2-52 a_2^2\right)}{(a_1-1)^2 (a_1+1)^2 \left(9 a_1^2 a_2^2-a_1^2-4 a_2^2\right)^2}\nonumber\\ 
&&+\frac{2 i a_1a_2\sigma_3 \left(63 a_1^6 a_2^2+a_1^6+215 a_1^4 a_2^2-3 a_1^4-90 a_1^2 a_2^2-2 a_1^2-8 a_2^2\right)}{(a_1-1)^2 (a_1+1)^2 \left(9 a_1^2 a_2^2-a_1^2-4 a_2^2\right)^2}\nonumber\\
&&-\frac{2 a_1 \sigma_1 \left(162 a_1^6 a_2^4-3 a_1^6 a_2^2+a_1^6+18 a_1^4 a_2^4+27 a_1^4 a_2^2-a_1^4-112 a_1^2 a_2^4-12 a_1^2 a_2^2+32 a_2^4\right)}{(a_1-1)^2 (a_1+1)^2 \left(9 a_1^2 a_2^2-a_1^2-4 a_2^2\right)^2}\nonumber\\
&&+\frac{\mathbb{I} \left(81 a_1^8 a_2^4-18 a_1^8 a_2^2+a_1^8+414 a_1^6 a_2^4+66 a_1^6 a_2^2-335 a_1^4 a_2^4-24 a_1^4 a_2^2-a_1^4+24 a_1^2 a_2^4+16 a_2^4\right)}{(a_1-1)^2 (a_1+1)^2 \left(9 a_1^2 a_2^2-a_1^2-4 a_2^2\right)^2}\nonumber\\
&&
\end{eqnarray}
we find that Eqn. \eqref{sonestwo.eq} is satisfied only  if
$a_1\in\{0,1,-1\}$. If $a_1=\pm 1$, then $\mathbb{I}-\hQ_1$ is not
invertible and if $a_1=0$, $\hQ_1$ is not invertible. Thus there are
no solutions. 

Similarly, if instead we pick $\hQ_4=a_4\sigma_2$ we again find no
solutions.  The following is the complete list of cases that we have
explicitly checked and found that none of them 
satisfy Eqn. \eqref{sonestwo.eq}. 
\begin{enumerate}
\item $\hQ_1=a_1\sigma_1 $, $\hQ_2=a_2\sigma_2$, $\hQ_3=a_3\sigma_1$ and $\hQ_4=a_4\sigma_2$
\item $\hQ_1=a_1\sigma_1 $, $\hQ_2=a_2\sigma_1$, $\hQ_3=a_3\sigma_2$ and $\hQ_4=a_4\sigma_1$
\item $\hQ_1=a_1\sigma_1 $, $\hQ_2=a_2\sigma_2$, $\hQ_3=a_3\sigma_2$ and $\hQ_4=a_4\sigma_2$
\item  $\hQ_1=a_1\sigma_1 $, $\hQ_2=a_2\sigma_1$, $\hQ_3=a_3\sigma_2$ and $\hQ_4=a_4\sigma_2$
\end{enumerate} 
Thus, we have found that a Pauli matrix representation of the
generators of the $\cQ$ subalgebra does
 not  extend to a representation of the full algebra $\cA$.

\bibliography{SSReferences} 
 \bibliographystyle{ieeetr} 
\end{document}